\documentclass{llncs}
\usepackage{paralist}
\usepackage{multirow}

\PassOptionsToPackage{usenames,dvipsnames}{xcolor}

\usepackage{graphicx}
\usepackage{lscape}

\usepackage{tikz}
\usepackage{amsmath,amssymb}
\usepackage{pgf}
\usepackage{pgflibraryarrows}
\usepackage{pgffor}
\usepackage{pgflibrarysnakes}
\usepackage{todonotes}
\usepackage{wasysym}
\usetikzlibrary{fit} % fitting shapes to coordinates
\usetikzlibrary{backgrounds} % drawing the background after the foreground
\usetikzlibrary{patterns}
\usetikzlibrary{positioning}
\tikzstyle{background}=[rectangle,fill=gray!10, inner sep=0.1cm, rounded corners=0mm]
\usepgflibrary{shapes}
\usetikzlibrary{snakes,automata}
\usetikzlibrary{shadows}
\usepackage[usenames,dvipsnames]{xcolor}
\tikzstyle{background}=[rectangle,fill=gray!10, inner sep=0.1cm, rounded corners=0mm]
\tikzstyle{loc}=[draw,rectangle,minimum size=1.4em,inner sep=0em]
\tikzstyle{trans}=[-latex, rounded corners]
\tikzstyle{trans2}=[-latex, dashed, rounded corners]
\tikzstyle{player2}=[draw,dashed,minimum size=5mm]
\makeatletter
\newif\if@restonecol
\makeatother

\usepackage[ruled,vlined,linesnumbered,lined]{algorithm2e}

\usepackage{xspace}
\usepackage{array}

\definecolor{lightgray}{gray}{0.9}

\DeclareMathOperator{\interior}{int}
\newcommand{\norm}[1]{\|#1\|}

\mathchardef\breakingcomma\mathcode`\,
{\catcode`;=\active
  \gdef;{\breakingcomma\discretionary{}{}{}}
}

\newcommand{\point}[1]{{{#1}}}

\newcommand{\px}{\point{x}}
\newcommand{\py}{\point{y}}

\newcommand{\vv}{\vec{v}}

\newcommand{\vb}{\vec{b}}

\newcommand{\vr}{\vec{r}}

\newcommand{\ball}[2]{B_{#1}(#2)}

\newcommand{\set}[1]{\left\{ #1 \right\}}

\newcommand{\seq}[1]{\langle #1 \rangle}
\newcommand{\Rplus}{{\mathbb R}_{\geq 0}}
\newcommand{\Nat}{\mathbb N}

\newcommand{\Real}{\mathbb R}

\newcommand{\Aa}{\mathcal{A}}
\newcommand{\Oo}{\mathcal{O}}

\newcommand{\Hh}{\mathcal{H}}

\newcommand{\Cc}{\mathcal{C}}

\newcommand{\Mm}{\mathcal{M}}
\newcommand{\Ww}{\mathcal{W}}

\usepackage{color}
\usepackage{pgfplots}
\usepgfplotslibrary{units}

\newcommand{\Ashu}[1]{}%{\color{blue} (${\bf AT:}$ {#1})}}
\newcommand{\Umang}[1]{}%\colorbox{NavyBlue}{\textcolor{Goldenrod}{#1}}}
\newcommand{\Umangm}[1]{}%\colorbox{OliveGreen}{#1}}

\newcommand{\Reach}{\textsc{Reach}}

\newcommand{\RUN}{\text{\it Run}}

\usepackage{scalefnt}
\newcommand\ScaleExists[1]{\vcenter{\hbox{\scalefont{#1}$\exists$}}}
\newcommand\ScaleForAll[1]{\vcenter{\hbox{\scalefont{#1}$\forall$}}}

\DeclareMathOperator*\bigexists{%
  \vphantom\sum
    \mathchoice{\ScaleExists{2}}{\ScaleExists{1.4}}{\ScaleExists{1}}{\ScaleExists{0.75}}}

\DeclareMathOperator*\bigforall{%
  \vphantom\sum
    \mathchoice{\ScaleForAll{2}}{\ScaleForAll{1.4}}{\ScaleForAll{1}}{\ScaleForAll{0.75}}}

\usepackage{framed}

\usepackage{caption}
\usepackage{subcaption}
\usepackage{hyperref}

\begin{document}

\title{The Reach-Avoid Problem for Constant-Rate Multi-Mode
  Systems\thanks{This research was supported in part by CEFIPRA project
    AVeRTS and by DARPA under agreement number FA8750-15-2-0096.
    All opinions stated are those of the authors and
    not necessarily of the organizations that have supported this research.}
  \thanks{The authors would like to thank the anonymous reviewers for their careful reading
    of the earlier versions of this manuscript
    and their many insightful comments and suggestions. }}
\author{
   Shankara Narayanan Krishna\inst{1} \and
   Aviral Kumar\inst{1} \and\\
   Fabio Somenzi\inst{2} \and
   Behrouz Touri\inst{2} \and
   Ashutosh Trivedi\inst{2}} 

 \institute{
   Indian Institute of Technology Bombay, India.
   \and 
   University of Colorado Boulder, USA.
 }
\maketitle

\begin{abstract}
  A constant-rate multi-mode system is a hybrid system that can switch freely
  among a finite set of modes, and whose dynamics is specified by a finite
  number of real-valued variables with mode-dependent constant rates. 
  Alur, Wojtczak, and Trivedi have shown that reachability problems for
  constant-rate multi-mode systems for open and convex safety sets can be 
  solved in polynomial time.
  In this paper we study the reachability problem for non-convex state spaces,
  and show that this problem is in general undecidable.
  We recover decidability by making certain assumptions about the safety set.
  We present a new algorithm to solve this problem and compare its performance
  with the popular sampling based algorithm rapidly-exploring random tree
  (RRT) as implemented in the Open Motion Planning Library (OMPL).
\end{abstract} 

\section{Introduction}
\label{sec:introduction}
Autonomous vehicle planning and control
frameworks~\cite{KTINTH15,apex} often follow the  hierarchical planning architecture
outlined by Firby~\cite{Firby89} and Gat~\cite{Gat98}.
The key idea here is to separate the complications involved in low-level
hardware control from high-level planning decisions to accomplish the navigation
objective. 
A typical example of such separation-of-concerns is proving the controllability
property (vehicle can be steered from any start point to arbitrary neighborhood
of the target point) of the motion-primitives of the vehicle followed by the search
(path-planning) for an obstacle-free path (called the \emph{roadmap}) and then
utilizing the controllability property to compose the low-level primitives to
follow the path (path-following).
However, in the absence of the controllability property, it is not always
possible to follow arbitrary roadmaps with given motion-primitives.
In these situations we need to study a motion planning problem that is
not opaque to the motion-primitives available to the controller.

We study this motion planning problem in a simpler setting of systems modeled as
constant-rate multi-mode systems~\cite{ATW12}---a switched system with
constant-rate dynamics (vector) in every mode---and study the reachability
problem for the non-convex safety sets. 
Alur et al.~\cite{ATW12} studied this problem for convex safety sets and 
showed that it can be solved in polynomial time.
Our key result is that even for the case when the safety set is defined using
polyhedral obstacles, the problem of deciding reachability is undecidable.
On a positive side we show that if the safety set is an open set
defined by linear inequalities, the problem is decidable and can be
solved using a variation of cell-decomposition algorithm~\cite{SS83}.
We present a novel bounded model-checking~\cite{clarke2001bounded} inspired
algorithm equipped with acceleration to decide the reachability. 
We use the Z3-theorem prover as the constraint satisfaction engine for the
quadratic formulas in our implementation. 
We show the efficiency of our algorithm by comparing its performance with
the popular sampling based algorithm \emph{rapidly-exploring
  random tree} (RRT)  as implemented in the \textit{Open Motion Planning
Library (OMPL)}.

For a detailed survey of motion planning algorithms we refer to the
 excellent expositions by Latombe~\cite{latombe2012robot} and
 LaValle~\cite{Lav06}.
The motion-planning problem while respecting system dynamics can be
modeled~\cite{frazzoli2000robust} in the framework of hybrid
automata~\cite{ACHH92,Hen96}; however the reachability problem is
undecidable even for simple stopwatch automata~\cite{HKPV98}.
There is a vast literature on decidable subclasses of hybrid
automata~\cite{ACHH92,BBM98}.
Most notable among these classes are initialized rectangular hybrid
automata~\cite{HKPV98}, two-dimensional piecewise-constant derivative
systems~\cite{AMP95}, timed automata~\cite{alurDill94}, and
discrete-time control for hybrid automata~\cite{HK99}. 
For a review of related work on multi-mode systems we refer to~\cite{AFMT13,ATW12}.

%%For lack of space proofs are either sketched or omitted; 
%%full proofs can be found in the technical report~\cite{KKSTT17}.

\section{Motivating Example}
\label{sec:motivation}
Let us consider a two-dimensional
multi-mode system with three modes $m_1, m_2$ and $m_3$ shown geometrically with
their rate-vectors in Figure~\ref{fig:l-shaped}(a).
We consider the reach-while-avoid problem in the arena given in
Figure~\ref{fig:l-shaped}(b) with two rectangular obstacles $\Oo_1$ and $\Oo_2$
and source and target points $\px_s$ and $\px_t$, respectively.
In particular, we are interested in the question whether it is possible to move
a point-robot from point $\px_s$ to point $\px_t$ using directions dictated by
the multi-mode system given in Figure~\ref{fig:l-shaped}(a) while avoiding
passing through or even grazing any obstacle.

It follows from our results in Section~\ref{sec:undec} that in general the
problem of deciding reachability is undecidable even with polyhedral obstacles.
However, the example considered in Figure~\ref{fig:l-shaped} has an interesting
property that the safety set can be represented as a union of finitely many
polyhedral open sets (cells). 
This property, as we show later, makes the problem decidable.
In fact, if we decompose the workspace  into cells using any off-the-shelf
cell-decomposition algorithm, we only need to consider the sequences of obstacle-free
cells to decide reachability.
In particular, for a given sequence of obstacle-free convex sets such
that the starting
point is in the first set, and the target point is in last set, one can write
a linear program checking whether there is a sequence of intermediate states,
one each in the intersection of successive sets, such that these points are
reachable in the sequence using the constant-rate multi-mode system. 
Our key observation is that one need not to consider cell-sequences larger than
the total number of cells since for reachability, it does not help for the system to
leave a cell and enter it again.

This approach, however, is not very efficient since one
needs to consider all sequences of the cells.
However, this result provides an upper bound on sequence of ``meta-steps'' or
``bound'' through the cells that system needs to take in order to reach the
target and hint towards a bounded model-checking~\cite{clarke2001bounded}
approach.
We progressively increase bound $k$  and ask whether there is a sequence of points
$\px_0, \ldots, \px_{k+1}$ such that $\px_0 = \px_s$, $\px_{k+1} = \px_t$, and
for all $0 \leq i \leq k$ we have that $\px_i$ can reach $\px_{i+1}$ using the
rates provided by the multi-mode system (convex cone of rates translated to
$\px_i$ contains $\px_{i+1}$) and the line segment $\lambda \px_i + (1-\lambda)
\px_{i+1}$ does not intersect any obstacle.
Notice that if this condition is satisfied, then the system can safely move from
point $\px_i$ to $\px_{i+1}$ by carefully choosing a scaling down of the
rates so as to stay in the safety set, as illustrated in Figure~\ref{fig:l-shaped}.

Let us first consider $k= 0$ and notice that one can reach point $\px_t$ from
$\px_s$ using just the mode $m_1$, however unfortunately the line segment
connecting these points passes through both obstacles.
In this case we increase the bound by $1$ and consider the problem of finding a
point $\px$ such that the system can reach from $\px_s$ to $\px$ and also from
$\px$ to $\px_t$, and the line segment connecting $\px_s$ with $\px$, and $\px$
with $\px_t$ do not intersect any obstacles.
It is easy to see from the Figure~\ref{fig:l-shaped} that it is indeed the
case. We can alternate modes $m_1, m_2$ from $x_s$ to $x$, and 
modes $m_1, m_3$ from $x$ to $x_t$. 
Hence, there is a schedule that steers the system from $\px_s$ to $\px_t$ as
shown in the Figure~\ref{fig:l-shaped}(c).
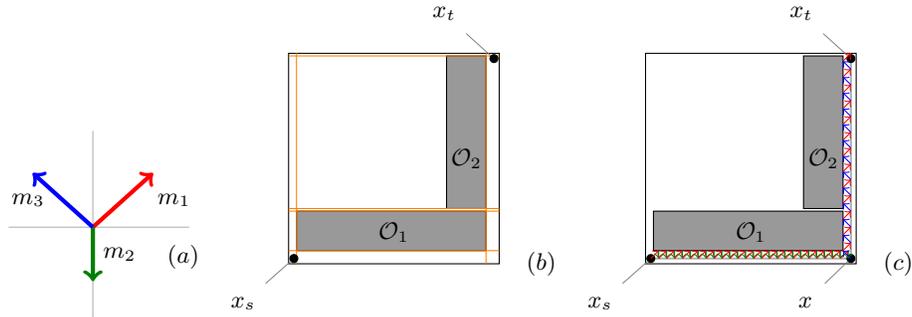
\begin{figure}[t]
  \centering
 
    \begin{tikzpicture}[scale=0.4]
    \tikzstyle{lines}=[draw=black!30,rounded corners]
    \tikzstyle{vectors}=[-latex, rounded corners]
    \tikzstyle{rvectors}=[-latex,very thick, rounded corners]
    \draw[lines] (4.2,0)--(10.2,0);
    \draw[lines] (7, 3.2)--(7,-3);
    \draw[->, ultra thick, red] (7, 0) --node[black, right]{$~~~m_1$} (9, 1.8) node[left]{$$};
    \draw[->, ultra thick, green!50!black] (7, 0) --node[black, right]{$m_2$} (7, -1.8) node[left]{$$};
    \draw[->, ultra thick, blue] (7, 0) --node[black, left]{$m_3~$} (5, 1.8) node[left]{$$};
    \node at (10, -1) {$(a)$};
    \end{tikzpicture}
    \hfill
    \begin{tikzpicture}[scale=0.7]
      \draw (0,0) rectangle (4, 4);
        
      \filldraw[fill=black!40!white, draw=black] (0.15, 1) rectangle(3.75, 0.25);
      \filldraw[fill=black!40!white, draw=black] (3, 3.95) rectangle(3.75, 1.05);
      
      \draw[-,orange](0,1) -- (4, 1);
      \draw[-,orange](0,0.25) -- (4, 0.25);
      \draw[-,orange](3.75,0) -- (3.75, 4);
      \draw[-,orange](0.15,0) -- (0.15, 4);
      
      \draw[-,orange](0, 1.05) -- (4, 1.05);
      \draw[-,orange](0,3.95) -- (4, 3.95);

      \node [fill, draw, circle, minimum width=3pt, inner sep=0pt, pin={[fill=white, outer sep=2pt]225:$\px_s$}] at (0.1,0.1) {}; 
      \node [fill, draw, circle, minimum width=3pt, inner sep=0pt, pin={[fill=white, outer sep=2pt]135:$\px_t$}] at (3.9,3.9) {};  
      
      \node at (2, .6) {$\Oo_1$};
      \node at (3.38, 2) {$\Oo_2$};
      \node at (4.8, 0) {$(b)$};
    \end{tikzpicture}
    \hfill
    \begin{tikzpicture}[scale=0.7]
	\draw (0,0) rectangle (4, 4);

        \filldraw[fill=black!40!white, draw=black] (0.15, 1) rectangle(3.75,
        0.25);
	\filldraw[fill=black!40!white, draw=black] (3, 3.95) rectangle(3.75, 1.05);
	 \draw [black!30, thick] (0.1, 0.1) -- (3.9, 0.1);
	 \draw [black!30, thick] (3.9, 0.1) -- (3.9, 3.9);    
	 \node [fill, draw, circle, minimum width=3pt, inner sep=0pt, pin={[fill=white, outer sep=2pt]225:$\px_s$}] at (0.1,0.1) {}; 
	 \node [fill, draw, circle, minimum width=3pt, inner sep=0pt, pin={[fill=white, outer sep=2pt]135:$\px_t$}] at (3.9,3.9) {};  
	 \node [fill, draw, circle, minimum width=3pt, inner sep=0pt,
           pin={[fill=white, outer sep=2pt]225: $\px$}] at (3.9,0.1) {};
         \def\x{0.15}
         \def\y{0.15}
         \def\z{0.15}
         \foreach \i in {0,...,24}{
           \draw[->,red](0.1+\z*\i,0.1)--(0.1+\x+\z*\i, 0.1+\y);
           \draw[->,green!50!black](0.1+\x+\z*\i,0.1+\y)--(0.1+\x+\z*\i, 0.1);
	 }
         
         \def\x{0.15}
         \def\y{0.15}
         \def\z{0.3}
         \foreach \i in {0,...,12}{
           \draw[->,blue](3.9,0.1+\z*\i)--(3.9-\x, 0.1+\z*\i+\y);
           \draw[->,red](3.9-\x, 0.1+\z*\i+\y)--(3.9, 0.1+\z*\i+2*\y);
	 }

         \node at (2, .6) {$\Oo_1$};
        \node at (3.38, 2) {$\Oo_2$};
         \node at (4.8, 0) {$(c)$};
\end{tikzpicture}
\caption{ a) A multi-mode system, b) an
  ``L''-shaped arena consisting of obstacles $\Oo_1$ and $\Oo_2$ with
  start and target points $\px_s$ and $\px_t$ along with the
  cell-decomposition shown by orange lines, and
  c) a safe schedule from $\px_s$ to $\px_t$. }
\label{fig:l-shaped}    
\end{figure}

The property  we need to check to ensure a safe schedule is the following: 
there exists a sequence of points $x_s=x_0,x_1,x_2,\dots, x_n=x_t$ such that 
for all $0\leq \lambda \leq 1$, and for all $i$, the line $\lambda x_i+(1-\lambda)x_{i+1}$ 
joining $x_i$ and $x_{i+1}$ does not intersect any obstacle $\mathcal{O}$. 
This can be thought of as a first-order formula of the
form $\exists X \forall Y F(X, Y)$ where $F(X, Y)$ is a linear formula.
By invoking the Tarski-Seidenberg theorem we know that checking the satisfiability
of this  property is decidable.
However, one can also give a direct quantifier elimination based on
Fourier-Motzkin elimination procedure to get existentially quantified quadratic
constraints that can be efficiently checked using theorem provers such as Z3
(\url{https://github.com/Z3Prover/z3}). 
This gives us a complete procedure to decide reachability for multi-mode systems
when the safety set can be represented as a union of finitely many polyhedral open sets.

\section{Problem Formulation}
\label{sec:problem}
\noindent{\bf Points and Vectors.} 
Let $\Real$ be the set of real numbers.
We represent the states in our system as points in $\Real^n$, which is equipped
with the standard \emph{Euclidean norm} $\norm{\cdot}$.
We denote points in this state space by $\px, \py$,  vectors by $\vr, \vv$, and 
the $i$-th coordinate of point $\px$ and vector $\vr$ by $\px(i)$ and $\vr(i)$,
respectively. 
The distance $\norm{\px, \py}$ between points $\px$ and $\py$ is defined as
$\norm{\px - \py}$. 

\noindent{\bf Boundedness and Interior.} 
We denote an {\em open ball} of radius $d \in \Rplus$ centered at $\px$ as
$\ball{d}{\px} {=} \set{\py {\in} \Real^n \::\: \norm{\px,\py} < d}$.
We denote a closed ball of radius $d \in \Rplus$ centered at $\px$ as
$\overline{\ball{d}{\px}}$. 
We say that a set $S \subseteq \Real^n$ is {\em bounded} if there exists
$d \in \Rplus$ such that, for all $\px, \py \in S$, we have
$\norm{\px,\py} \leq d$.
The {\em interior} of a set $S$, $\interior(S)$, is the set of all points
$\px \in S$, for which there exists $d > 0$ s.t. $\ball{d}{\px} \subseteq S$.

\noindent{\bf Convexity.} A point $\px$ is a \emph{convex
  combination} of a finite set of points $X = \set{\px_1, \px_2, \ldots, \px_k}$ if
there are $\lambda_1, \lambda_2, \ldots, \lambda_k \in [0, 1]$ such that
$\sum_{i=1}^{k} \lambda_i = 1$ and $\px = \sum_{i=1}^k \lambda_i \cdot \px_i$.
We say that $S \subseteq \Real^n$ is {\em convex} iff, for all
$\px, \py \in S$ and all $\lambda \in [0,1]$, we have
$\lambda \px + (1-\lambda) \py \in S$ and moreover,
$S$ is a {\em convex polytope} if there exists $k \in \Nat$, a
matrix $A$ of size $k \times n$ and a vector $\vb \in \Real^k$ such that $\px
\in S$ iff $A\px \leq \vb$.
A closed \emph{hyper-rectangle} is a convex polytope that can
be characterized as $\px(i) \in [a_i, b_i]$ for each $i \leq n$ where $a_i, b_i
\in \Real$.

\begin{definition}
  \label{def:BMMS}
  A (constant-rate) multi-mode system (MMS) is a tuple $\Hh = (M, n, R)$ where: 
    $M$ is a finite nonempty set of \emph{modes}, 
    $n$ is the number of continuous variables, and 
    $R : M \to \Real^n$ maps to each mode a rate vector
    whose $i$-th entry specifies the change in the value of the $i$-th
    variable per time unit.
    For computation purposes, we assume that the real numbers
    are rational.
\end{definition}

\begin{example}
  An example of a 2-dimensional multi-mode system  $\Hh = (M, n, R)$ is shown in
Figure~\ref{fig:l-shaped}(a) where $M = \set{ m_1, m_2, m_3}$, $n = 2$, and the
rate vector is such that $R(m_1) = (1, 1)$, $R(m_2) = (0, -1)$, and $R(m_3) =
(-1, 1)$. 
\end{example}
A \emph{schedule} of an MMS specifies a timed sequence of mode switches.
Formally, a \emph{schedule} is defined as a finite or infinite sequences of
\emph{timed actions}, where a timed action $(m, t) \in M \times \Rplus$ is a
pair consisting of a mode and a time delay.
A finite \emph{run} of an MMS $\Hh$ is a finite sequence of states and timed
actions $r = \seq{\px_0, (m_1, t_1), \px_1,  \ldots, (m_k, t_k), \px_k}$
such that for all $1 \leq i \leq k$ we have that
$\px_i = \px_{i-1} + t_i \cdot R(m_i)$.
For such a run $r$ we say that $\px_0$ is the \emph{starting state}, while
$\px_k$ is its \emph{terminal state}.
An \emph{infinite run} of an MMS $\Hh$ is similarly defined to be an infinite
sequence $\seq{\px_0, (m_1, t_1), \px_1, (m_2, t_2), \ldots}$ such that for all
$i \geq 1$ we have that $\px_i = \px_{i-1} + t_i \cdot  R(m_i)$.

Given a finite schedule $\sigma = \seq{(m_1, t_1), (m_2, t_2), \ldots, (m_k,
  t_k)}$ and a state $\px$, we write $\RUN(\px, \sigma)$ for the (unique)
finite run  $\seq{\px_0, (m_1, t_1), \px_1, (m_2, t_2), \ldots,
  \px_k}$ such that $\px_0 = \px$.
In this case, we also say that the schedule $\sigma$ steers the MMS $\Hh$ from the state
$\px_0$ to the state $\px_k$.

We consider the problem of MMS reachability within a given \emph{safety set}
$S$.
We specify the safety set by a pair $(\Ww, \Oo)$, where
$\Ww \subseteq \Real^n$ is called the \emph{workspace} and
$\Oo = \set{\Oo_1, \Oo_2, \ldots, \Oo_k}$ is a finite set of
\emph{obstacles}.
In this case the safety set $S$ is characterized as $S_{\Ww \backslash \Oo}  = \Ww \setminus \Oo$.
 We assume in the rest of the
paper that $\Ww = \Real^n$ and for all $1 \leq i \leq k$, $\Oo_i$ is a
\emph{convex} (not necessarily closed) polytope specified by a set of linear
inequalities. 
%Notice that the set $S$ is open if $\Ww$ is open and each $\Oo_i$ is closed.

We say that a finite run
$\seq{\px_0, (m_1, t_1), \px_1, (m_2, t_2),  \ldots}$ is $S$-safe if for
all $i \geq 0$ we have that $\px_i \in S$ and
$\px_i + \tau_{i+1} \cdot R(m_{i+1}) \in S$ for all $\tau_{i+1} \in [0, t_{i+1}]$.
Notice that if $S$ is a convex set then for all $i \geq 0$,
$\px_i \in S$ implies that for all $i \geq 0$ and for all $\tau_{i+1} \in
[0, t_{i+1}]$ we have that $\px_i + \tau_{i+1} \cdot  R(m_{i+1}) \in S$.
We say that a schedule $\sigma$ is $S$-safe from a state $\px$, or is $(S,
\px)$-safe, if the
corresponding unique run $\RUN(\px, \sigma)$ is $S$-safe.
Sometimes we simply call a schedule or a run safe when the safety set and
the starting state are clear from the context.
We say that a state $\px'$ is $S$-safe reachable from a state $\px$ if there
exists a finite schedule $\sigma$  that is $S$-safe at $\px$ and steers the
system from state $\px$ to $\px'$.

We are interested in solving the following problem.

\begin{definition}[Reachability]
  Given a constant-rate multi-mode system $\Hh = (M, n, R)$, safety
  set $S$, start state $\px_s$, and target state $\px_t$, the
  reachability problem $\Reach(\Hh, S_{\Ww \backslash \Oo}, \px_s, \px_t)$ is to decide
  whether there exists an $S$-safe finite schedule that steers the
  system from state $\px_s$ to $\px_t$.
\end{definition}

Alur \emph{et al.}~\cite{ATW12} gave a polynomial-time algorithm to
decide if a state $\px_t$ is $S$-safe reachable from a state $\px_0$
for an MMS $\Hh$ for a convex safety set $S$.  In particular, they
characterized the following necessary and sufficient condition.
\begin{theorem} [\cite{ATW12}]
  \label{thmcms}
  Let $\Hh = (M, n, R)$ be a multi-mode system and let $S \subset \Real^n$ be an
  open, convex safety set. 
  Then, there is an $S$-safe schedule from $\px_s \in S$ to $\px_t \in S$, if and only if there is $\vec{t} \in \Rplus^{|M|}$
  satisfying:
  $  \px_s + \sum_{i=1}^{|M|} R(m_i) \cdot \vec{t}(i) = \px_t$. 
\end{theorem}
A key property of this result is that if $\px_t$ is reachable from $\px_s$ without
considering the safety set, then it is also reachable inside arbitrary convex
set as long as both $\px_s$ and $\px_t$ are strictly in the interior of the
safety set.

We study the extension of this theorem for the reachability problem with
non-convex safety sets.  A key contribution of this paper is a precise
characterization of the decidability of the reachability problem for
multi-mode systems.
\begin{theorem}
  \label{thm:main}
  Given a constant-rate multi-mode system $\Hh$, workspace
  $\Ww = \Real^n$, obstacles set $\Oo$, start state $\px_s$ and target
  state $\px_t$, the reachability problem
  $\Reach(\Hh, S_{\Ww \setminus \Oo}, \px_s, \px_t)$ is in general
  undecidable.
  However, if the obstacle set $\Oo$ is given as finitely many closed
  polytopes, each defined by a finite set of linear inequalities, then
  reachability is decidable.
\end{theorem}

\section{Decidability}
\label{sec:decidability}
We prove the decidability condition of Theorem \ref{thm:main} in this section.
\begin{theorem}
\label{thm:dec}
For a MMS $\Hh = (M, n, R)$, a safety set $S$, a start state $\px_s$,
and a target state $\px_t$, the problem $\Reach(\Hh, S_{\Ww \backslash \Oo}, \px_s, \px_t)$
is decidable if $\Oo$ is given as finitely many closed polytopes.
\end{theorem}
For the rest of this section let us fix  a MMS $\Hh =
(M, n, R)$,
a start state $\px_s$ and a target state $\px_t$.  Before we prove
this theorem, we define cell cover (a notion related to, but distinct
from the one of cell decomposition introduced
in~\cite{latombe2012robot}).

\begin{definition}[Cell Cover]
  Given a safety set $S \in \Real^n$, a cell of $S$ is an open, convex
  set that is a subset of $S$.  A \emph{cell cover} of $S$ is a
  collection $\Cc = \set{c_1, \ldots, c_N}$ of cells whose union
  equals $S$.  Cells $c, c' \in \Cc$ are \emph{adjacent} if and only
  if $c \cap c'$ is non-empty.
\end{definition}

A \emph{channel} in $S$ is a finite sequence
$\seq{c_1, c_2, \ldots, c_N}$ of cells of $S$ such that $c_i$ and
$c_{i+1}$ are adjacent for all $1 \leq i < N$.  It follows that
$\cup_{1 \leq i \leq N} c_i$ is a path-connected open set.
A $\Cc$-channel is a channel whose cells are in cell cover
$\Cc$.

Given a channel $\pi = \seq{c_1, \ldots, c_N}$, a multi-mode
system $\Hh = (M, n, R)$, start and target states
$\px_s, \px_t \in S$, we say that $\pi$ is a \emph{witness} to
reachability if the following linear program is feasible:
\begin{gather}
  \bigexists_{0 \leq i \leq N} \px_i  \,.
  \Big( \px_s = \px_0  \wedge  \px_t = \px_N \Big) \wedge
  \Big(1 \leq i < N \rightarrow \px_i \in (c_i \cap c_{i+1})\Big)
  \,\wedge  \label{eq4} \\
  \bigexists_{1 \leq i \leq N, m \in M} t_i^{(m)} \,. \Big(t_i^{(m)}
  \geq 0\Big) \wedge \bigwedge_{1
    \leq i \leq N} \Big(\px_i = \px_{i-1} + \sum_{m\in M} R(m) \cdot  
    t_i^{(m)}\Big) \enspace.\nonumber
\end{gather}

\begin{lemma}
  \label{witness}
  If $S$ is an open safety set, there exists a finite $S$-safe
  schedule that solves $\Reach(\Hh, S, \px_s, \px_t)$ if and only if
  $S$ contains a witness channel $\seq{c_1, c_2, \ldots, c_N}$ for some
  $N \in \mathbb{N}$.
\end{lemma}
\begin{proof}
  ($\Leftarrow$) If $\seq{c_1, c_2, \ldots, c_N}$ is a witness channel,
  then for $0 < i \leq N$, $\px_{i-1}$ and $\px_i$ are in $c_i$.
  Theorem~\ref{thmcms} guarantees the existence of a $c_i$-safe
  schedule for each $i$.  The concatenation of these schedules is a
  solution to $\Reach(\Hh, S, \px_s, \px_t)$.

  ($\Rightarrow$)
  The run of a finite schedule that solves
  $\Reach(\Hh, S, \px_s, \px_t)$ defines a closed, bounded subset $P$ of $S$.
  Since $S$ is open, every point $x \in P$ is contained in a cell
  of $S$.
  Collectively, these cells form an open cover of $P$.
  By compactness, then, there is a finite subcover of $P$.
  If any element of the subcover is entered by the run more than once, there
  exists another run that is contained in that cell between the first entry and
  the last exit.
  For such a run, if two elements of the subcover are entered at the same time,
  the one with the earlier exit time is redundant.
  Therefore, there is a subcover in which no two elements are entered by
  the run of the schedule at the same time.
  This subcover can be ordered according to the time at which the run enters
  each cell to produce a sequence that satisfies the definition of witness
  channel. 
  %% The run of a finite schedule that solves
  %% $\Reach(\Hh, S, \px_s, \px_t)$ defines a closed, bounded subset $P$
  %% of $S$.  Since $S$ is open, every point $\px \in P$ is contained in
  %% a cell of $S$.  Collectively, these cells form an open cover of $P$.
  %% By compactness, then, there is a finite subcover of $P$; in
  %% particular, there is one in which no two elements are entered by the
  %% run of the schedule at the same time.  This subcover can be ordered
  %% according to the time at which the run enters each cell to produce a
  %% sequence that satisfies the definition of witness channel. 
  \qed
\end{proof}

\begin{lemma}
  \label{generic-witness-to-c-witness}
  If $S$ is an open safety set and $\Cc$ a cell cover of
  $S$, there exists a witness channel for $\Reach(\Hh, S, \px_s, \px_t)$
  iff there exists a witness $\Cc$-channel.
\end{lemma}
\begin{proof}
  One direction is obvious.  Suppose therefore that there exists a
  witness channel; let $\sigma$ be the finite schedule whose existence is
  guaranteed by Lemma~\ref{witness}.  The path that is traced in the MMS $\Hh$ 
  when steered by $\sigma$ is a bounded closed subset $P$ of $S$
  because it is the continuous image of a compact interval of the real
  line.  (The time interval in which $\Hh$ moves from $\px_s$ to
  $\px_t$.)  Since $\Cc$ is an open cover of $P$, there exists a
  finite subset of $\Cc$ that covers $P$; specifically, there is an
  irredundant finite subcover such that no two cells are entered at
  the same time during the run of $\sigma$.  This subcover can be
  ordered according to entry time to produce a sequence of cells that
  satisfies the definition of witness channel. \qed
\end{proof}

\begin{lemma}
\label{cell-decomposition-to-cover}
  If $\Oo$ is a finite set of closed polytopes, then a finite cell
  cover of the safety set $S$ is computable.
\end{lemma}
\begin{proof}
  If $\Oo$ is a finite set of closed polytopes, one can apply the the
  vertical decomposition algorithm of \cite{latombe2012robot} to
  produce a cell \emph{decomposition}.  Each cell $C$ in this
  decomposition of dimension less than $n$ that is not contained in
  the obstacles (and hence is entirely contained in $S$) is replaced
  by a convex open set obtained as follows.  Let $B$ be an
  $n$-dimensional box around a point of $C$ that is in $S$.  The
  desired set is the convex hull of the set of vertices of either $C$
  or $B$. \qed
\end{proof}

\begin{proof}[of Theorem~\ref{thm:dec}]
  Lemmas~\ref{witness}--\ref{generic-witness-to-c-witness}
  imply that $\Reach(\Hh, S, \px_s, \px_t)$ is decidable if a finite
  cell cover of $S$ is available.  If $\Oo$ is given as a finite set
  of closed polytopes, each presented as a set of linear inequalities,
  then Lemma~\ref{cell-decomposition-to-cover} applies.  \qed
\end{proof}

\begin{algorithm}[t]
  \KwIn{MMS $\Hh = (M, n, R)$, two points $\px_s, \px_t$, workspace $\Ww$,
    obstacle set $\Oo$, and an upper bound $B$ on number of cells in a cell-cover. }
  \KwOut{NO, if no safe schedule exists and otherwise
    such a schedule.}
  \BlankLine
  $k \leftarrow 0$;
\While{$k \leq B $}{
    Check if the following formula is satisfiable:
    \begin{eqnarray*}
     \bigexists\limits_{1 \leq i \leq N} \px_i && \bigexists\limits_{1 \leq i \leq N, m \in M}
     t_i^{(m)} \text{ s.t. }
      \left( \px_s = \px_1  \wedge  \px_t = \px_N \right) \wedge   
      \bigwedge\limits_{\stackrel{1\leq i \leq N}{m \in M}}  t_i^{(m)} \geq 0  \wedge  \\
      & 
      \bigwedge\limits_{i=2}^{N} & \left(\px_i = \px_{i-1} + \sum_{m\in M} R(m) \cdot
      t_i^{(m)}\right)  \wedge 
       \bigwedge\limits_{i=2}^{N} \textsc{ObstacleFree}(\px_{i-1}, \px_{i})
    \end{eqnarray*}
    \lIf{not satisfiable}{
      $k \leftarrow k+1$
    }
    \uElse
        {
          Let $\sigma$ be an empty sequence\;
          \For{$i = 1$ to $k-1$ }{
            $\sigma = \sigma :: \textsc{Reach\_Convex}(\Hh, \px_i, \px_{i+1}, S)$}
        }
        \Return $\sigma$\;
  }
  \caption{\textsc{BoundedMotionPlan}($\Hh, \Ww, \Oo, \px_s, \px_t, B$)}
  \label{alg:mainLoop}
\end{algorithm}

\begin{algorithm}[h!]
\caption{\label{algo:mms} $\textsc{Reach\_Convex}(\Hh, \px_s, \px_t, S)$}
\KwIn{MMS $\Hh = (M, n, R)$, two points $\px_s, \px_t$, convex, open, safety set $S$} 
\KwOut{NO if no $S$-safe schedule from $\px_s$ to $\px_t$ exists and otherwise
  such a schedule.}
$t_1 =
\min\limits_{m \in M} \max \set{\tau \::\: \px_s + \tau \cdot R(m) \in S}$\;
$t_2 =
\min\limits_{m \in M} \max \set{\tau \::\: \px_t + \tau \cdot R(m) \in S}$\;
$t_\text{safe} = \min \set{t_1, t_2}$\;
Check whether the following linear program is feasible: 
\begin{align}
  \px_s + \sum_{m\in M} R(m) \cdot t^{(m)} = \px_t
  & \text{ and } &
  t^{(m)} \geq 0 \text{ for all $m \in M$}~\label{eqn3}
\end{align}

\lIf{no satisfying assignment exists}{\Return NO}
\Else{
  Find an assignment $\{t^{(m)}\}_{m\in M}$.

  Set $l = \lceil (\sum_{m \in M} t^{(m)})/t_\text{safe}\rceil $.

  {\bf return} the following schedule  $\seq{(m_i, t_i)}$ where 
  \[
  m_k = (k \bmod |M|) + 1 \text{  and } t_k = t^{(m_k)} / l
  \text{  for $k = 1, 2, \ldots, l|M|$.}
  \]
}
\end{algorithm}

The algorithm implicit in the proof of Theorem~\ref{thm:dec} requires
one to compute the cell cover in advance, and enumerate
sequences of cells in order to decide reachability.  We next present
an algorithm inspired by bounded model
checking~\cite{clarke2001bounded} that implicitly enumerates sequences
of cells of increasing length till the upper bound on number of cells
is reached, or a safe schedule from the source point to the target
point is discovered.  The key idea is to guess a sequence of points
${\px_1, \ldots, \px_N}$ starting from the source point and ending in
the target point such that for every $1 \leq i < N$ the point
$\px_{i+1}$ is reachable from $\px_i$ using rates provided by the
multi-mode system.  Moreover, we need to check that the line segment
connecting $\px_i$ and $\px_{i+1}$ does not intersect with obstacles,
i.e:
%\[
$\bigforall_{0 \leq \lambda \leq 1} (\lambda \px_i + (1-\lambda) \px_{i+1}) \not \in \cup_{j = 1}^{k} \Oo_j$.
%\]
We write $\textsc{ObstacleFree}(\px_i, \px_{i+1})$ for this condition.
Algorithm~\ref{alg:mainLoop} sketches a bounded-step algorithm to
decide reachability for multi-mode systems that always terminates for
multi-mode systems with sets of closed obstacles defined by linear
inequalities thanks to Theorem~\ref{thm:dec}.

Notice that at line~$2$ of algorithm \ref{alg:mainLoop}, we need to check the feasibility of
the constraints system, which is of the form $\exists X \forall Y F(X, Y)$ where
universal quantifications are implicit in the test for $\textsc{ObstacleFree}$.
If the solver we use to solve the constraints has full support to solve the
$\forall$ quantification, we can use that to solve the above constraint. In our
experiments, we used the Z3 solver (\url{https://github.com/Z3Prover/z3})
to implement the Algorithm~\ref{alg:mainLoop} and found that the
solver was unable to solve in some cases.
Fortunately, the universal quantification in our constraints is of very special
form and can be easily removed using the Fourier-Motzkin elimination procedure,
which results in quadratic constraints that are efficiently solvable by Z3
solver.
In Section~\ref{sec:experiments} we present the experimental results on some
benchmarks to demonstrate scalability. % of our algorithm.

\section{Undecidability}
\label{sec:undec}
In this section we give a sketch of the proof of the following undecidability
result. 
\begin{theorem}
  \label{thm:undec}
  Given a constant-rate multi-mode system $\Hh$, convex workspace
  $\Ww$, obstacles set $\Oo$, start state $\px_s$ and target state $\px_t$, 
  the reachability problem $\Reach(\Hh, S_{\Ww \backslash \Oo}, \px_s, \px_t)$ is 
  in general undecidable. 
\end{theorem}
\begin{proof}{(Sketch.)}
We prove the undecidability of this problem by giving a reduction from
the halting problem for two-counter machines that is known to be
undecidable~\cite{Min67}.
Given a two counter machine $\Aa$ having instructions $L = \ell_1$, $\dots$,
$\ell_{n-1}, \ell_{halt}$,  we construct a multi-mode system $\Hh_\Aa$ along
with  non-convex safety $S_{\Ww \backslash \Oo}$ characterized using linear constraints.
The idea is to simulate the unique run of two-counter machine $\Aa$ via the unique
safe schedule of the MMS $\Hh_\Aa$ by going through a sequence of modes such
that a pre-specified target point is reachable iff the counter machine halts.  

\noindent \textbf{Modes}. For every increment/decrement instruction $\ell_i$ of the counter machine we
have two modes $\Mm_i$ and $\Mm_{ik}$, where  $k$ is the index of the unique
instruction $\ell_k$ to which the control shifts in $\Aa$ from $\ell_i$.
For every zero check instruction $\ell_i$,  we have four modes $\Mm^1_i,
\Mm^2_i, \Mm_{ik}$ and $\Mm_{im}$, where  $k,m$ are respectively the indices of
the unique instructions $\ell_k, \ell_m$ to which the control shifts from
$\ell_i$ depending on whether the counter value is $>0$ or $=0$. 
There are three modes $\Mm_{halt}, \Mm^{c_1}_{halt}$ and $\Mm^{c_2}_{halt}$ 
corresponding to the halt instruction. 
We have a special ``initial'' mode $\mathcal{I}$ which is the first mode to be
applied in any safe schedule.

\noindent \textbf{Variables}. The MMS $\Hh_\Aa$ has two variables $C=\{c_1, c_2\}$ that store the value
of two counters.  
There is a unique variable $S=\{s_0\}$ used to enforce that mode  $\mathcal{I}$
as the first mode. 
For every increment or decrement instruction $\ell_i$, 
there are variables $w_{ij}, x_{ij}$, where $j$ is the index of the unique
instruction 
$\ell_j$ to which control shifts from $\ell_i$. 
We define variable $z_{i\#}$ for each zero-check instruction $\ell_i$. 

\noindent \textbf{Simulation}.   A simulation of the two counter machine going through instructions 
$\ell_0$, $\ell_1$, $\ell_2$, $\dots, \ell_y, \ell_{halt}$ is achieved by going through modes
$\mathcal{I}, \Mm_{0}, \Mm_{01}, \Mm_1$ or $\Mm^1_1$ or $\Mm^2_1 \dots, \Mm_y, \Mm_{y~halt}$ in order, spending exactly one unit of time in each mode. 
Starting from a point $\px_s$ with $s_0=1$ and $v=0$ for all variables $v$ other than $s_0$, we 
want to reach a point $\px_t$  where $w_{halt}=1$ and 
 $v=0$ for all variables $v$ other than $w_{halt}$. 
The idea is to start in  mode $\mathcal{I}$, and  
spending one unit of time in  $\mathcal{I}$ obtaining $s_0=0, w_{01}=1$
(spending a time other than one violates safety, see Lemma~\ref{one}). 
Growing $w_{01}$ represents that the current instruction is $\ell_0$, and the next one is $\ell_1$. 
Next, we shift to mode $\Mm_0$, 
spend one unit of time there to obtain $x_{01}=1, w_{01}=0$.
This is followed by 
mode $\Mm_{01}$, where $x_{01}$ becomes 0, and one of the variables $z_{1\#}, w_{12}$ attain 1, depending on whether 
$\ell_1$ is a zero check instruction or not (again, spending a time other than
one in $\Mm_0, \Mm_{01}$ violates safety, see Lemma~\ref{two}).

In general, while at a  mode $\Mm_{ij}$, the next instruction 
$\ell_k$ after $\ell_j$ is chosen  by ``growing'' the variable $w_{jk}$ if $\ell_j$ is not a zero-check instruction, or 
by ``growing'' the  variable $z_{j\#}$  if $\ell_j$ is a zero-check instruction.
 In parallel, $x_{ij}$ grows down to 0, so that $x_{ij}+w_{jk}=1$ or $x_{ij}+z_{j\#}=1$. 
 The sequence of choosing modes, and enforcing that one unit of time be spent in each mode 
 is necessary to adhere to the safety set as can be seen by Lemmas~\ref{two} and~\ref{three}. 
 \begin{itemize}
\item  In the former case, the control shifts from $\Mm_{ij}$ to mode 
$\Mm_j$ where variable $x_{jk}$ grows at rate 1 while $w_{jk}$ grows at rate -1, so that 
 $x_{jk}+w_{jk}=1$. Control shifts from $\Mm_j$ to $\Mm_{jk}$, where the next instruction $\ell_g$ after $\ell_k$ is chosen 
 by growing  variable $w_{kg}$ if $\ell_k$ is not zero-check instruction, or 
 the variable $z_{k\#}$ is grown if $\ell_k$ is a zero-check instruction.  
  \item  In the latter case, one of the modes $\Mm^1_j,\Mm^2_j$ is chosen from $\Mm_j$ where 
   $z_{j\#}$ grows at rate -1. 
  Assume $\ell_j$ is the instruction ``If the counter value is $>0$, then goto $\ell_m$, else goto $\ell_h$".   
      If  $\Mm^1_j$ is chosen, then the variable $x_{jm}$ grows at rate 1
   while if $\Mm^2_j$ is chosen, then the variable $x_{jh}$ grows at rate 1. 
   In this case, we have $z_{j\#}+x_{jm}=1$ or $z_{j\#}+x_{jh}=1$. 
   From $\Mm^1_j$, control shifts to $\Mm_{jm}$, while from 
   $\Mm^2_j$, control shifts to $\Mm_{jh}$.
\end{itemize}
Continuing in the above fashion, we eventually reach mode $\Mm_{y~halt}$
where $x_{y~halt}$ grows down to 0, while the variable $w_{halt}$ grows to 1, so that 
$x_{y~halt}+w_{halt}=1$(see Lemma~\ref{last} which enforces this). 

Starting from $\px_s$---which lies in the hyperplane $H_0$
given as  $s_0+w_{0j}=1$  where $\ell_j$ is the unique instruction  following
$\ell_0$---a safe execution stays in $H_0$ as long as control stays in the initial mode
$\mathcal{I}$.  
Control then switches to mode $\Mm_0$,  to the hyperplane $H_1$ given by 
 $w_{0j}+x_{0j}=1$.  Note that $H_0 \cap H_1$ is non-empty and intersect at the
point   where $w_{0j}=1$, and all other variables are $0$.
Spending a unit of time at $\Mm_0$, control switches to mode $\Mm_{0j}$, and  
to the hyperplane $H_2$ given by $x_{0j}+w_{jk}=1$ depending on whether
$\ell_j$ is not a zero-check instruction.
Again, note that $H_1 \cap H_2$ is non-empty and intersect  
at the point where $c_1=1, x_{0j}=1$ and all other variables are zero.
This continues, and we obtain a safe transition from hyperplane $H_i$ to
$H_{i+1}$ as dictated by the simulation of the two counter machine.
The sequence of safe hyperplanes lead to the hyperplane $H_{last}$ given by
$w_{halt}=1$ and all other variables 0 iff the two counter machine halts.
%% REMOVE BELOW
Appendix \ref{app:undec-eg} 
gives an example of a reduction 
from 2-counter machines. 
\qed
\end{proof}

\section{Experimental Results}
\label{sec:experiments}
In this section, we discuss some preliminary results obtained with an
implementation of Algorithm~\ref{alg:mainLoop}.
In order to show competitiveness of the proposed algorithm, we compare its
performance with a popular implementation of the RRT algorithm~\cite{Lav06} on a
collection of micro-benchmarks (some of these benchmarks are inspired by \cite{saha}).

\subsection{Experimental Setup}
Rapidly-exploring Random
Tree (RRT)~\cite{Lav06} is a space-filling data structure that is used to search a
region by incrementally building a tree.
It is constructed by selecting random points in the state space and can provide
better coverage of reachable states of a system than mere simulations.
There are many versions of RRTs available;  we use the \emph{Open Motion Planning
Library (OMPL)} implementation of RRT for our experiments. The
OMPL  library (http://ompl.kavrakilab.org)
consists of many
state-of-the-art, sampling-based motion planning algorithms. We used the RRT API provided by the OMPL library.
The results for RRT were obtained with a goal bias parameter set to $0.05$, and obstacles implemented 
as \texttt{StateValidityCheckerFunction()} as mentioned in the documentation~\cite{sucan2012the-open-motion-planning-library}. 

We implemented our algorithm on the top of the Z3
solver~\cite{DeMoura:2008:ZES:1792734.1792766}. The
implementation involves coding formulae in FO-logic over reals and checking
for a satisfying assignment. Our algorithm was implemented in Python 2.7. The OMPL implementation was done in C++.  
The experiments with Algorithm~\ref{alg:mainLoop}
and RRT were performed on a computer running Ubuntu 14.10,
with an Intel Core i7-4510 2.00 GHz quadcore CPU, with 8 GB RAM. 
We compared the two algorithms by executing them on a
set of microbenchmarks whose
 obstacles are hyper-rectangular, though our
algorithm can handle general polyhedral obstacles.
\begin{figure}[t] 
    \begin{center}
%  \scalebox{0.9}{
    \begin{tikzpicture}
      \draw (0,0) rectangle (7, 4);
      \filldraw[fill=black!40!white, draw=black] (1, 0) rectangle(2, 3.5);
      \filldraw[fill=black!40!white, draw=black] (2.5, 4) rectangle(3.5, 0.5);
      \filldraw[fill=black!40!white, draw=black] (4, 0) rectangle(5, 3.5);
      \filldraw[fill=black!40!white, draw=black] (5.5, 4) rectangle(6.5, 0.5);
      \node [fill, draw, circle, minimum width=3pt, inner sep=0pt, pin={[fill=white, outer sep=2pt]225:$\px_s$}] at (0.2,0.1) {}; 
      \node [fill, draw, circle, minimum width=3pt, inner sep=0pt, pin={[fill=white, outer sep=2pt]135:$\px_t$}] at (6.9,3.9) {}; 
      \draw [black!30, thick] (0.2, 0.1) -- (0.2, 3.9); 
      \node [fill, draw, circle, minimum width=3pt, inner sep=0pt, pin={[fill=white, outer sep=2pt]135:$\px_1$}] at (0.2,3.9) {};
      \draw [black!30, thick] (0.2, 3.9) -- (2.2, 3.9);
      \node [fill, draw, circle, minimum width=3pt, inner sep=0pt, pin={[fill=white, outer sep=2pt]135:$\px_2$}] at (2.2,3.9) {};
      \draw [black!30, thick] (2.2, 0.2) -- (2.2, 3.9);
      \node [fill, draw, circle, minimum width=3pt, inner sep=0pt, pin={[fill=white, outer sep=2pt]225:$\px_3$}] at (2.2,0.2) {};
      \draw [black!30, thick] (2.2, 0.2) -- (3.7, 0.2);
      \node [fill, draw, circle, minimum width=3pt, inner sep=0pt, pin={[fill=white, outer sep=2pt]225:$\px_4$}] at (3.7,0.2) {};
      \draw [black!30, thick] (3.7, 3.9) -- (3.7, 0.2);
      \node [fill, draw, circle, minimum width=3pt, inner sep=0pt, pin={[fill=white, outer sep=2pt]135:$\px_5$}] at (3.7,3.9) {};
      \draw [black!30, thick] (3.7, 3.9) -- (5.2, 3.9);
      \node [fill, draw, circle, minimum width=3pt, inner sep=0pt, pin={[fill=white, outer sep=2pt]135:$\px_6$}] at (5.2,3.9) {};
      \draw [black!30, thick] (5.2, 0.2) -- (5.2, 3.9);
      \node [fill, draw, circle, minimum width=3pt, inner sep=0pt, pin={[fill=white, outer sep=2pt]225:$\px_7$}] at (5.2,0.2) {};
      \draw [black!30, thick] (6.9, 0.2) -- (5.2, 0.2);
      \node [fill, draw, circle, minimum width=3pt, inner sep=0pt, pin={[fill=white, outer sep=2pt]225:$\px_8$}] at (6.9,0.2) {};
      \draw [black!30, thick] (6.9, 0.2) -- (6.9, 3.9);
      
      \node at (1.5, 1.75) {$\Oo_1$};
      \node at (3, 1.75) {$\Oo_2$};
      \node at (4.5, 1.75) {$\Oo_3$};
      \node at (6, 1.75) {$\Oo_4$};
    \end{tikzpicture}
\begin{tikzpicture}
\draw (0,0) rectangle (4, 4);

\filldraw[fill=black!40!white, draw=black](0.5, 3) rectangle (3.5, 3.5);
\filldraw[fill=black!40!white, draw=black](0.5, 1) rectangle (1, 3);
\filldraw[fill=black!40!white, draw=black](0.5, 0.5) rectangle (3.5, 1);
\filldraw[fill=black!40!white, draw=black](1.25, 1.25) rectangle (3.5, 1.5);
\filldraw[fill=black!40!white, draw=black](1.25, 2.5) rectangle (3.5, 2.75);
\filldraw[fill=black!40!white, draw=black](3.25, 1.5) rectangle (3.5, 2.5);
\filldraw[fill=black!40!white, draw=black](2, 1.65) rectangle (3, 1.85);
\filldraw[fill=black!40!white, draw=black](2, 2.15) rectangle (3, 2.3);
\filldraw[fill=black!40!white, draw=black](2, 1.85) rectangle (2.15, 2.15); 

\node [fill, draw, circle, minimum width=3pt, inner sep=0pt, pin={[fill=white, outer sep=2pt]225:$\px_s$}] at (0.1,0.1) {}; 
\draw [black!30, thick] (0.1, 0.1) -- (3.9, 0.1);
\node [fill, draw, circle, minimum width=3pt, inner sep=0pt, pin={[fill=white, outer sep=2pt]225:$\px_1$}] at (3.9,0.1) {}; 

\draw [black!30, thick] (3.9, 0.1) -- (3.9, 2.85);
\node [fill, draw, circle, minimum width=3pt, inner sep=0pt, pin={[fill=white, outer sep=2pt]355:$\px_2$}] at (3.9,2.85) {};
\draw [black!30, thick] (3.9, 2.85) -- (1.15, 2.85);
\node [fill, draw, circle, minimum width=3pt, inner sep=0pt, pin={[outer sep=-2pt]120:$\px_3$}] at (1.15,2.85) {};

\draw [black!30, thick] (1.15, 1.6) -- (3.075, 1.6);
\node [fill, draw, circle, minimum width=3pt, inner sep=0pt, pin={[outer sep=-3pt]225:$\px_4$}] at (1.15,1.6) {};

\draw [black!30, thick] (1.15, 1.6) -- (1.15, 2.85);
\node [fill, draw, circle, minimum width=3pt, inner sep=0pt, pin={[outer sep=-3pt]355:$\px_5$}] at (3.075,1.6) {};

\draw [black!30, thick] (3.075, 1.6) -- (3.075, 2);
\node [fill, draw, circle, minimum width=3pt, inner sep=0pt, pin={[outer sep=-3pt]355:$\px_6$}] at (3.075,2) {};

\draw [black!30, thick] (3.075, 2) -- (2.3, 2);
\node [fill, draw, circle, minimum width=3pt, inner sep=0pt, pin={[outer sep=-3pt]180:$\px_t$}] at (2.3,2) {};
\end{tikzpicture}
 \begin{tikzpicture}
  \draw (0,0) rectangle (4, 4);
  
  \filldraw[fill=black!40!white, draw=black] (0.15, 1) rectangle(3.75, 0.25);
  \filldraw[fill=black!40!white, draw=black] (3, 3.95) rectangle(3.75, 1.05);
  
  \node [fill, draw, circle, minimum width=3pt, inner sep=0pt, pin={[fill=white, outer sep=2pt]180:$\px_s$}] at (2.85,3.7) {}; 
  \node [fill, draw, circle, minimum width=3pt, inner sep=0pt, pin={[fill=white, outer sep=2pt]0:$\px_t$}] at (3.85,3.7) {};  
  
  \node at (2, .6) {$\Oo_1$};
  \node at (3.38, 2) {$\Oo_2$};
  \node at (0, -0.6) {};
  \node at (-1, 0) {};
 \end{tikzpicture}
% }
\end{center}

\caption{a) Snake-shaped arena with four obstacles (left), b) Maze-shaped
  arena with three $C$-shaped patterns (middle) and c) modified L-shaped
  arena (right).}
  
\label{fig:modified_Lshaped}
\label{fig:maze-shaped} 
\end{figure}
\begin{table}[t]
\begin{center}
\begin{tabular}{ | c | c | c | c | c | c |}
\hline
\multirow{2}{*}{\textbf{Dimension}} & \multirow{2}{*}{\textbf{Arena Size}} &
\multicolumn{2}{| c |}{~~~~~\textbf{OMPLRRT}~~~~~} & \multicolumn{2}{| c |}{~~~~~\textbf{BoundedMotionPlan}~~~~~} \\
\cline{3-6}
 & & Time(s) & Nodes & Time(s) & Witness Length \\

\hline 
2 & 100 $\times$ 100 & 0.011 & 8 & 0.012 & 2 \\
\hline
2 & 1000 $\times$ 1000 & 0.076 & 245 & 0.012 & 2 \\
\hline 
3 & 100 $\times$ 100 &  0.107 & 4836 & 0.183 & 2 \\
\hline
3 & 1000 $\times$ 1000 & 1.9 & 1800 & 0.19 & 2 \\
\hline
4 & 100 $\times$ 100 & 1.2 & 612 & 0.201 & 2 \\ 
\hline 
4 & 1000 $\times$ 1000 & 94.39 & 2857 & 0.206 & 2 \\
\hline
5 & 100 $\times$ 100 & 3.12 & 778 & 2.69 & 2 \\
\hline
5 & 1000 $\times$ 1000 &  149.4 & 2079 & 2.68 & 2 \\
\hline
6 & 1000 $\times$ 1000 & 105 & 3822 & 15.3 & 2 \\
\hline
7 & 1000 $\times$ 1000 & 319.63 & 2639 & 190.3 & 2 \\
\hline  
\end{tabular}
\end{center}
\caption{Summary of results for the $L$ shaped arena}
\label{ellShaped}
\end{table}
We considered the following microbenchmarks.
\begin{itemize}
\item \textbf{L-shaped arena.}
      This class of microbenchmarks contains  examples with hyper-rectangular
  workspace and certain ``L'' shaped obstacles as shown in 
  Figure~\ref{fig:l-shaped}.
  The initial vertex is the lower left vertex of the square ($\px_s$) and the
  target is the right upper vertex of the square ($\px_t$).
  Our algorithm can give the solution to this problem with bound $B = 2$
  returning the sequence $\seq{\px_1, \px, \px_t}$ as shown in the figure,
  while the \textsc{Rrt} algorithm in this case samples most of the points which lie
  on the other side of the obstacles and if the control modes are not in the
  direction of the line segments $\px_1 \px$ and $\px\px_t$, 
  then it grows in arbitrary directions and hits the obstacles a large number of
  times, leading to a large number of iterations slowing the growth.
  We experimented with L-shaped examples for dimensions ranging from $2$ to
  $7$. In most of the cases, we found that the performance of
  the \textsc{BoundedMotionPlan} algorithm was better than that of
  OMPLRRT. Another important point to note is that RRT or other simulation-based
  algorithms do not perform well as the input size increases, which can be
  clearly seen from the running times obtained on increasing arena
  sizes in Table~\ref{ellShaped}. Our algorithm worked better than RRT
  for higher dimensions ($\geq$ 3).

\item \textbf{Snake-shaped arena.}
  The name comes from the serpentine appearance of the safe sets in these arenas.
  The motivation to study these microbenchmarks comes from motion planning
  problems in regular environments.
  The arena has rectangular obstacles coming from the top and the
  bottom (as shown in Figure~\ref{fig:maze-shaped} for two dimensions)
  alternately.
  The starting point is the lower left vertex $\px_s$ and the target point is
  $\px_t$.
  A sample free-path through the arena is also shown in the figure.
  \textsc{Rrt} algorithm performs well for lower dimensions but fails to
  terminate for higher dimensions. 
  The results for this class of obstacles are summarised in Table~\ref{snakeShaped}. Experiments
  were performed for up to 3 dimensions and 4 obstacles.   

\item \textbf{Maze-shaped arena.} 
  These benchmarks mimic the motion planning situations where the task of the
  robot is to navigate through a maze.
  We model a maze using finitely many concentric ``C''-shaped obstacles with
  different orientations as shown
  in Figure~\ref{fig:maze-shaped}.
  The task is to navigate from the lower left outer corner to the
  center point of the square.
  This kind of arena seems to be particularly challenging for the RRT algorithm
  and the growth of the tree seems to be quite slow.
  Also, the performance of our tool degrades as the bound increases due to a
  increase in the number of constraints, and hence, these examples
  require more time as compared to the other two microbenchmarks.
  However, as shown in Table~\ref{mazeShaped}, \textsc{OmplRrt} and
  \texttt{BoundedMotionPlan}
  perform almost equally well, with the latter being slightly better.
%%%%%%%%%%%%%%%-----------------------$$$$$$$$$$$$$$$$$
\item \textbf{Modified L-shaped obstacles.}
  These set of microbenchmarks contains a hyperrectangular workspace and 2
  hyperrectangular obstacles arranged in a ``L-shaped'' fashion as shown in
  Figure~\ref{fig:modified_Lshaped}. The initial vertex lies very close to one
  of the obstacles. The target vertex is the vertex very close to the start
  vertex but on the other side of the obstacle. Our algorithm can give the
  solution to this problem with bound $B = 3$ while \textsc{Rrt} algorithm
  spends time in sampling from the bigger obstacle-free part of the arena. The
  results are  summarised in Table~\ref{modifiedEll}. 
\end{itemize}

\begin{table}[t]
\begin{center}
\begin{tabular}{ | c | c | c | c | c | c | c | c |}
\hline
\multirow{2}{*}{\textbf{Dim.}} & \multirow{2}{*}{\textbf{Arena Size}} & \multirow{2}{*}{\textbf{Obstacles}} & \multicolumn{3}{| c |}{\textbf{OMPLRRT}} & \multicolumn{2}{| c |}{\textbf{BoundedMotionPlan}} \\ \cline{4-8}
 & & & Time(s) & Nodes &  Nodes in Path & Time(s) & Witness Length \\
 \hline
 2 & 350 $\times$ 350 & 3 & 3.56 & 13142 & 72 & 2.54 & 4 \\
 \hline 
 2 & 350 $\times$ 350 & 4 & 4.12 & 15432 & 96 & 4.23 & 5 \\
 \hline 
 2 & 3500 $\times$ 3500 & 3 & 4.79 & 15423 & 83 & 2.57 & 4 \\
 \hline
 3 & 350 $\times$ 350 & 3 & 102.3 & 86314 & 67 & 96.43 & 4 \\
 \hline
 3 & 3500 $\times$ 3500 & 3 & 100.22 & 1013 & 27 & 96.42 & 4 \\
 \hline 
\end{tabular}
\end{center}
\caption{Summary of results for the snake-shaped arena}
\label{snakeShaped}
\end{table} 

\begin{table}[t]
\begin{center}
\begin{tabular}{ | c | c | c | c | c | c | c | c | }
\hline
\multirow{2}{*}{\textbf{Dim.}} & \multirow{2}{*}{\textbf{Arena Size}} & \multirow{2}{*}{\textbf{Obstacles}} & \multicolumn{3}{| c |}{\textbf{OMPLRRT}} & \multicolumn{2}{| c |}{\textbf{BoundedMotionPlan}} \\ \cline{4-8}
& & & Time(s) & Nodes &  Nodes in Path & Time(s) & Bound \\
\hline
2 & $600 \times 600$ & 2 & 1.8 & 9500 & 60 & 1.3 & 4 \\
\hline
2 & $6000 \times 6000$ & 3 & 23.5 & 11256 & 78 & 45.23 & 5 \\
\hline
3 & $600 \times 600$ & 2 & 132.6 & 90408 & 71 & 120.3 & 5 \\
\hline
3 & $6000 \times 6000$ & 3 & 1002.6 & 183412 & 93 & 953.4 & 5 \\
\hline  
\end{tabular}
\end{center}
\caption{Summary of results for the maze-shaped arena}
\label{mazeShaped}
\end{table}

\begin{table}[h]
  \begin{center}
    \centering
\begin{tabular}{ | c | c | c | c | c | c | c | }
\hline
\textbf{Dimension } & \textbf{Arena Size} & \multicolumn{3}{| c |}{\textbf{OMPLRRT}} & \multicolumn{2}{| c |}{\textbf{BoundedMotionPlan}} \\
\hline
& & Time & Nodes &  Nodes in Path & Time & Bound \\
\hline
2 & $100 \times 100$ & 0.445 & 27387 & 40 & 0.126 & 3 \\
\hline
2 & $1000 \times 1000$ & 2.57 & 38612 & 47 & 0.132 & 3 \\
\hline
3 & $100 \times 100$ & 115.23 & 57645 & 71 & 92.1 & 3 \\
\hline
3 & $1000 \times 1000$ & 675.62 & 183412 & 93 & 95.23 & 3 \\
\hline
4 & $100 \times 100$ & 287.32 & 64230 & 65 & 283.23 & 3 \\
\hline
4 & $1000 \times 1000$ & 923.45 & 192453 & 78 & 292.53 & 3 \\
\hline
5 & $100 \times 100$ & 523.62 & 73422 & 69 & 534.45 & 3 \\
\hline
5 & $1000 \times 1000$ & 1043 & 223900 & 72 & 533.96 & 3\\
\hline  
\end{tabular}
\end{center}
\caption{Summary of results for the modified L-shaped obstacles}
\label{modifiedEll}
\end{table}

The micro-benchmarks presented above involved the situations where
the target point is reachable from the source point.
It is interesting to see the performance of two algorithms in cases when there
is no path from the source to target point.
For the cases when an upper bound on cell-decomposition can be imposed,
our algorithm is capable of producing negative answer.
Table~\ref{negtab} summarizes the performance of \textsc{OmplRrt}
and \texttt{BoundedMotionPlan} for $L$-shaped arenas when the target
point is not reachable.  The timeout for
RRT was set to be 500 seconds, and it did not terminate until the timeout,
which is as expected. On the other hand, \textsc{BoundedMotionPlan} performed well, with running times close to those when the target point is reachable.
\begin{table}[t]
\begin{center}
\centering
\begin{tabular} { | c | c | c | c |}
\hline 
\multirow{2}{*}{~~~~\textbf{Dimension}~~~~} & \multicolumn{2}{| c |}{~~~~~~~~~~\textbf{OMPLRRT}~~~~~~~~~} & {~~~~\textbf{BoundedMotionPlan}~~~~} \\ \cline{2-4}
& Time(s) & Nodes & Time(s) \\ \hline
2 & 500 (TO) & 5301778 & 0.0088 \\ \hline
3 & 500 (TO) & 7892122 & 0.032 \\ \hline
4 & 500 (TO) & 4325621 & 0.056 \\ \hline
5 & 500 (TO) & 5624609 & 2.73 \\ \hline
6 & 500 (TO) & 4992951 & 18.34 \\ \hline
7 & 500 (TO) & 3765123 & 213.23 \\ \hline

\end{tabular}
\end{center}
\caption{Summary of results for the unreachable L-shaped obstacles.}
\label{negtab}
\end{table}

\noindent\textbf{Discussion}.
Our implementation of \textsc{BoundedMotionPlan} even though
preliminary, compares favorably with a state-of-the-art implementation
of \textsc{Rrt}. \textsc{BoundedMotionPlan}, in addition, can naturally deal
with restrictions on the dynamics of the MMS, that is, with systems
such that the positive linear span of the mode vectors is not
$\Real^n$.

A trend observed in our experiments is that if a large fraction of the
arena is covered by obstacles, then the probability of a randomly sampled
point lying in the obstacle region is high and this makes RRT ineffective in
this situation by wasting a lot of iterations.
Another trend is that as the arena size increases, it becomes more
difficult for RRT to navigate to the destination points even with higher values
of goal bias.

Our algorithm performs better in situations when it terminates early
(target reachable from source with shorter witnesses) while
the performance of our algorithm degrades as the bound or the
dimensions increases since the number of constraints introduced by the
Fourier-Motzkin like-procedure implemented in our algorithm grows
exponentially with the dimension exhibiting the curse of
dimensionality.

\section{Conclusion}
\label{sec:conclusion}
In this paper we studied the motion planning problem for constant-rate
multi-mode system with non-convex safety sets given as a convex set of
obstacles.
We showed that while the general problem is already undecidable in this simple
setting of linearly defined obstacles, decidability can be recovered by making
appropriate assumption on the obstacles. 
Moreover, our algorithm performs satisfactorily when compared to well-known
algorithms for motion planning, and can easily be adapted to provide
semi-algorithms for motion-planning problems for objects with polyhedral
shapes.
While the algorithm is complete for classes of safety sets for which a
bound on the size of a cell cover can be effectively computed,
bounds based on cell decompositions of the safety set may be too large to
be of practical use.  This situation is akin to that encountered in
bounded model checking of finite-state systems, in which bounds based
on the radii of the the state graph are usually too large.  We are
therefore motivated to look at extensions of the algorithm that
incorporate practical termination checks.

\newpage
\appendix
\section{Proof of Theorem~\ref{thm:undec}} %Undecidability Proof}
\label{app:undec}
In this section we present details of the proof of our undecidability theorem.
We show undecidability of the motion planning problem by giving a reduction from
the undecidable halting problem for two-counter machines.

\begin{definition}
A \emph{two-counter machine} (Minsky machine) $\Aa$ is a tuple $(L, C)$
where:
${L = \set{\ell_0, \ell_1, \ldots, \ell_{n-1}, \ell_{halt}}}$ is the set of
instructions and  ${C = \set{c_1, c_2}}$ is the set of two \emph{counters}. 
There is a distinguished terminal instruction  $\ell_{halt}$ called
HALT and the instructions $L$ are one of the following types:
\begin{itemize}
\item 
  {\textbf{increment}}.
  $\ell_i : c := c+1$;  goto  $\ell_k$,
\item
  {\textbf{decrement}}.
  $\ell_i : c := c-1$;  goto  $\ell_k$,
\item
  {\textbf{zero-test}}.
  $\ell_i$ : if $(c{>}0)$ then goto $\ell_k$
  else goto $\ell_m$,
\item
  {\textbf{Halt}}.
  $\ell_{halt}:$ HALT.
\end{itemize}
where $c \in C$, $\ell_i, \ell_k, \ell_m \in L$.
Let $I, D,$ and $O$ represent the sets of increment, decrement and  
zero-check instructions, respectively. 
\end{definition}
 
A configuration of a two-counter machine is a tuple $(\ell, c, d)$ where
$\ell \in L$ is an instruction, and $c, d$ are natural numbers that specify the
value
of counters $c_1$ and $c_2$, respectively.
The initial configuration is $(\ell_0, 0, 0)$.
A run of a two-counter machine is a (finite or infinite) sequence of
configurations $\seq{k_0, k_1, \ldots}$ where $k_0$ is the initial
configuration, and the relation between subsequent configurations is
governed by transitions between respective instructions.
The run is a finite sequence if and only if the last configuration is
the terminal instruction $\ell_{halt}$.
Note that a two-counter  machine has exactly one run starting from the initial
configuration. We assume without loss of generality that $\ell_0$ is 
an increment instruction.
The \emph{halting problem} for a two-counter machine asks whether 
its unique run ends at the terminal instruction $\ell_{halt}$.
It is well known~\cite{Min67} that the halting problem for two-counter machines is
undecidable.

\subsection{Reduction}  
Given a two counter machine $\Aa$ having instructions $L = \ell_1$, $\dots$,
$\ell_{n-1}, \ell_{halt}$,  we construct a MMS $\Hh_\Aa$ having a number of modes and
variables polynomial in  $n$. 
The idea is to simulate the two-counter machine in the MMS by going through a
sequence of modes such that a target point is reachable iff the counter machine
halts.  
%We will give a brief sketch of our construction.
We will next present the details of our reduction by characterizing the set of
modes, the set  of variables, rates of variables in different modes, as well as
the set of obstacles for the instance of multi-mode system corresponding to a
given  instance of counter machine.  
\begin{itemize}
\item {\bf Modes.}
  For every increment/decrement instruction $\ell_i \in I \cup D$,
we have two modes $\Mm_i$ and $\Mm_{ik}$, where 
$k$ is the index of the unique instruction $\ell_k$ to which the control shifts 
in $\Aa$ from $\ell_i$.
For every zero check instruction $\ell_i \in O$, 
we have four modes $\Mm^1_i, \Mm^2_i, \Mm_{ik}$ and $\Mm_{im}$, where 
$k,m$ are respectively the indices of the unique instructions 
$\ell_k, \ell_m$ to which the control shifts from $\ell_i$ depending on whether 
the counter value is $>0$ or $=0$. 
There are three modes $\Mm_{halt}, \Mm^{c_1}_{halt}$ and
   $\Mm^{c_2}_{halt}$ 
corresponding to the halt instruction. 
We have a special ``initial'' mode $\mathcal{I}$ which is the first mode to be
applied in any safe schedule in our reduction. This property is ensured using a
special variable $s_0$ and a careful definition of obstacles. 

\item {\bf Variables.}
  The multi-mode system has two variables $C=\{c_1, c_2\}$ that store the value
  of two counters.  
  There is a unique variable $S=\{s_0\}$, whose rate is 0 at all modes other than
  the  mode $\mathcal{I}$ and is used to enforce that mode  $\mathcal{I}$ is the
  first valid mode in the starting state. 
  For  $\ell_i \in I \cup D$, 
  there are variables $w_{ij}, x_{ij}$, where $j$ is the index of the unique
  instruction 
  $\ell_j$ to which control shifts from $\ell_i$.
  We write $W$ for the set
  $\set{w_{ij}, w_{i~halt} \::\: 0 \leq i,j \leq n \text{ and } \ell_i \notin O }$
  and 
  $X=\set{x_{ij} \::\: 0 \leq i,j \leq n}$.
  Also, we define a variable $z_{i\#}$ for every $\ell_i \in O$ and we write
  $Z=\set{z_{i\#} \::\: \ell_i \in O}$. 
  Hence the set of variables is
  \[
  \mathcal{V} = X \cup W \cup Z \cup C \cup S \cup \set{w_{halt}}.
  \]
  Let $X_{halt}$ be the subset of $X$ consisting of variables of the form $x_{i~halt}$.  
  Observe that the dimension of the MMS $|\mathcal{V}|$ is $\mathcal{O}(n^2)$, where $n$ is the number of instructions 
in the two counter machine.

\item {\bf Intution for Dynamics and Obstacles.}
  A simulation of the two counter machine going through instructions 
$\ell_0$, $\ell_1$, $\ell_2$, $\dots, \ell_y, \ell_{halt}$ is achieved by going through modes
$\mathcal{I}, \Mm_{0}, \Mm_{01}, \Mm_1$ or $\Mm^1_1$ or $\Mm^2_1 \dots, \Mm_y, \Mm_{y~halt}$ in order, spending exactly one unit of time in each mode. 
Starting from a point with $s_0=1$ and $v=0$ for all variables $v$ other than $s_0$, we 
want to reach a point where $w_{halt}=1$ and 
 $v=0$ for all variables $v$ other than $w_{halt}$. 
The idea is to start in  mode $\mathcal{I}$, and  
spending one unit of time in  $\mathcal{I}$ obtaining $s_0=0, w_{01}=1$. 
Growing $w_{01}$ represents that the current instruction is $\ell_0$, and the next one is $\ell_1$. 
Next, we shift to mode $\Mm_0$, 
spend one unit of time there to obtain $x_{01}=1, w_{01}=0$. 
This is followed by 
mode $\Mm_{01}$, where $x_{01}$ becomes 0, and one of the variables $z_{1\#}, w_{12}$ attain 1, depending on whether 
$\ell_1$ is a zero check instruction or not. 

In general, while at a  mode $\Mm_{ij}$, the next instruction 
$\ell_k$ after $\ell_j$ is chosen  by ``growing'' the variable $w_{jk}$ if $\ell_j$ is not a zero-check instruction, or 
by ``growing'' the  variable $z_{j\#}$  if $\ell_j$ is a zero-check instruction.
 In parallel, $x_{ij}$ grows down to 0, so that $x_{ij}+w_{jk}=1$ or $x_{ij}+z_{j\#}=1$. 
\begin{itemize}
\item  In the former case, the control shifts from $\Mm_{ij}$ to mode 
$\Mm_j$ where variable $x_{jk}$ is grows at rate 1 while $w_{jk}$ grows at rate -1, so that 
 $x_{jk}+w_{jk}=1$. Control shifts from $\Mm_j$ to $\Mm_{jk}$, where the next instruction $\ell_g$ after $\ell_k$ is chosen 
 by growing  variable $w_{kg}$ if $\ell_k$ is not zero-check instruction, or 
 the variable $z_{k\#}$ is grown if $\ell_k$ is a zero-check instruction.  
  \item  In the latter case, one of the modes $\Mm^1_j,\Mm^2_j$ is chosen from $\Mm_j$ where 
   $z_{j\#}$ grows at rate -1. 
  Assume $\ell_j$ is the instruction ``If the counter value is $>0$, then goto $\ell_m$, else goto $\ell_h$".   
      If  $\Mm^1_j$ is chosen, then the variable $x_{jm}$ grows at rate 1
   while if $\Mm^2_j$ is chosen, then the variable $x_{jh}$ grows at rate 1. 
   In this case, we have $z_{j\#}+x_{jm}=1$ or $z_{j\#}+x_{jh}=1$. 
   From $\Mm^1_j$, control shifts to $\Mm_{jm}$, while from 
   $\Mm^2_j$, control shifts to $\Mm_{jh}$.
\end{itemize}
Continuing in the above fashion, we eventually reach mode $\Mm_{y~halt}$
where $x_{y~halt}$ grows down to 0, while the variable $w_{halt}$ grows to 1, so that 
$x_{y~halt}+w_{halt}=1$. It remains to  use the modes $\Mm_{halt}, \Mm^{c}_{halt}$ 
as many times to obtain $c_1=0, c_2=0$ and $w_{halt}=1$. 

\item {\bf Dynamics.}
  We will next define the rates of the variables in different modes.
\begin{enumerate}
\item Variable rates at mode $\mathcal{I}$ are
 such that $R(\mathcal{I})(s_0)=-1, R(\mathcal{I})(w_{0j})=1$, while $R(\mathcal{I})(v)=0$
for all variables $v$ other than $s_0$ and $w_{0j}$. Here $j$ is the index of
the unique instruction $\ell_j$ to which the control shifts from the initial
instruction $\ell_0$ (recall that $\ell_0$ is an increment instruction, and
hence control shifts deterministically to some $\ell_j$). 
 
\item  
Assume $\ell_i \in I \cup D$ is an increment/decrement instruction for counter $c_1$($c_2$) and 
 let $\ell_j$ be the resultant instruction. 
The rates of variables at mode $\Mm_i$ are
$R(\Mm_i)(w_{ij})=-1, R(\Mm_i)(x_{ij})=1$,  and 
$R(\Mm_i)(v)=0$ for $v \neq c_1$ ($c_2$), while 
$R(\Mm_i)(c_1)=1$ ($R(\Mm_i)(c_2)=1$) if $\ell_i \in I$ and 
$R(\Mm_i)(c_1)=-1$ ($R(\Mm_i)(c_2)=-1$) if $\ell_i \in D$.  

\item For $\ell_i \in O$, the rates of variables in the modes $\Mm^1_i, \Mm^2_i, \Mm_{ik}, \Mm_{im}$ are
\begin{enumerate}
\item $R(\Mm^1_i)(z_{i\#})=-1,R(\Mm^1_i)(x_{ik})=1$, and we have $R(\Mm^1_i)(v)=0$ for all
  other variables $v$. 
\item $R(\Mm^2_i)(z_{i\#})=-1,R(\Mm^2_i)(x_{im})=1$, and we have $R(\Mm^2_i)(v)=0$ for all
  other variables $v$.   
  \end{enumerate}
 
 \item We have the modes $\Mm_{ij}$ for $i \in I \cup D \cup O$. 
The rates of variables at mode $\Mm_{ij}$, $j \neq halt$ are 
\begin{enumerate}
\item $R(\Mm_{ij})(x_{ij})=-1$
\item If $\ell_j$ is not a zero check instruction, 
then there is a  unique instruction $\ell_k$ to which the control will shift
from $\ell_j$.
Then $R(\Mm_{ij})(w_{jk})=1$, while $R(\Mm_{ij})(v)=0$ for all  variables
$v \notin \{w_{jk}, x_{ij}\}$.
\item If $\ell_j$ is a zero check instruction, then we have
 $R(\Mm_{ij})(z_{j\#})=1$, while $R(\Mm_{ij})(v)=0$ for all  variables
$v \notin \{z_{j\#}, x_{ij}\}$. 
\end{enumerate}
\item
  The rates of variables at mode $\Mm_{i ~halt}$ are as follows:
  \[
  R(\Mm_{i halt})(x_{i~halt})=-1 \text{ and } R(\Mm_{i~halt})(w_{halt})=1
  \]
  while all other variables have rate $0$.
  The rates at modes 
  $\Mm_{halt}$, $\Mm^{c_1}_{halt}$, $\Mm^{c_2}_{halt}$ are given by :
  \begin{itemize}
  \item $R(\Mm_{halt})(w_{halt})=-1$ and $R(\Mm_{halt})(v)=0$ for all other
    variables.
  \item $R(\Mm^{c_1}_{halt})(c_1)=-1$, $R(\Mm^{c_1}_{halt})(w_{halt})=1$ while 
    $R(\Mm^{c_1}_{halt})(v)=0$ for all other variables. 
  \item $R(\Mm^{c_2}_{halt})(c_2)=-1$, $R(\Mm^{c_2}_{halt})(w_{halt})=1$ while
    $R(\Mm^{c_2}_{halt})(v)=0$ for all other variables.
  \end{itemize}
\end{enumerate}
\item \textbf{Workspace and Obstacles.} 
Instead of describing the obstacles directly, we 
describe its complement. i.e. the \emph{safety set}.
The safety set consists of points in the
Euclidean space  satisfying $\bigwedge_{t=a}^g\varphi_t$ where: 
\begin{enumerate}
\item[($\varphi_a$)] $Init$:
  \begin{eqnarray*}
    0 \leq y \leq 1 & & \text{ for $y \in N \cup S \cup X \cup Z$}\\ 
    y \geq 0 && \text{ for $y \in \set{n_{halt}, c_1, c_2}$} 
  \end{eqnarray*}
\item[($\varphi_b$)] $Mutex(X)$ :
  \[
  \bigwedge_{i,j} \left(x_{ij}>0 \Rightarrow \bigwedge_{\tiny{\begin{array}{c}(g \neq
        i) \vee \\ (f \neq j) \end{array}}} x_{gf}=0 \right)
  \]
\item[($\varphi_c$)] $Mutex(W, Z)$ : 
  \[
  \bigwedge_{i,j} \left(w_{ij}>0 \Rightarrow \left(\bigwedge_{\tiny{\begin{array}{c}(g \neq
        i) \vee \\ (f \neq j) \end{array}}} w_{gf}=0  \wedge
  \bigwedge_{g}z_{g\#}=0 \right)\right)
  \]
\item[($\varphi_d$)] $Mutex(Z, W)$: 
  \[
  \bigwedge_{i} \left(z_{i\#}>0 \Rightarrow \left(\bigwedge_{k \neq i} z_{k\#}=0  \wedge
  \bigwedge_{g,f} w_{gf}=0 \right)\right)
  \]
\item[($\varphi_e$)] $Mutex(S,X)$:
  \[
  s_0 >0 \Rightarrow \bigwedge_{i,j} x_{ij}=0
  \]
\item[($\varphi_f$)] $Mutex(w_{halt}, X \cup W \cup Z \cup S)$:
  \[
  w_{halt}>0 \Rightarrow \bigwedge_{y \in \mathcal{V}\backslash (X_{halt} \cup
    C)}(y=0)
  \]
\item[($\varphi_g$)] $Sum(X_{halt},w_{halt})$:
  \[
  \bigwedge_i  \left(x_{i~halt} >0 \Rightarrow \left(x_{i~halt} + w_{halt}=1\right)\right)
  \]
\end{enumerate}

An obstacle $\mathcal{O}$ is thus one which satisfies $\bigvee_{t=a}^g \neg \varphi_t$.
As an example, $\mathcal{O}_{i~halt}=x_{i~halt} >0 \wedge (x_{i~halt}+w_{halt} \neq 1)$ is an obstacle 
obtained by negating $\varphi_g$.
Note that the safety set thus defined is not necessarily an open set. 
\end{itemize}

\subsection{Correctness of the Reduction}
We represent a point in the state space of the multi-mode system as a tuple of
valuation to all variables with an arbitrary but fixed ordering. 
In our ordering, $s_0$ is the first variable, followed by $c_1$ and $c_2$ as the next
two variables, and $w_{halt}$ as the last variable in the tuple.  
For the multi-mode system constructed earlier, we show that 
starting from the initial point $(1,0,0, \dots, 0)$ it is possible to safely
reach the target point $(0,0,0, \dots, 0,1)$ if and only if the corresponding
two counter machine halts.  

We present a set of lemma to prove the correctness. In particular, our lemmas establish that 
\begin{itemize}
\item The schedule begins in mode $\mathcal{I}$, and exactly one unit of time
  is spent in $\mathcal{I}$ (Lemma \ref{one}), 
\item the order of choosing modes is decided by the sequence of instructions 
  chosen in a correct simulation of the two counter machine (Lemmas \ref{two}, \ref{three}, \ref{last}),
  and
\item starting from $s_0=1$ and $v=0$ for all variables $v \neq s_0$, one can reach the point with $w_{halt}=1$ and 
  $v=0$ for all variables $v \neq w_{halt}$ avoiding obstacles iff (1), (2) are
  true. This is shown by Lemma \ref{really-last}.
\end{itemize}

\begin{lemma}
\label{one}
Any safe schedule must begin in mode $\mathcal{I}$ and exactly one unit of time
is spent at this mode. 
\end{lemma}
\begin{proof}
  Observe that the starting point is such that $s_0=1$ and all other variables
  are 0.
  Assume if the schedule begins in a mode other than $\mathcal{I}$ 
  and spend $t$ units of time there.
  If the schedule begins in some mode $\Mm_i$ and spend time $t$, then we will
  obtain  $w_{ij}=-t$ for some $j$, hitting the obstacle $\neg Init$  (violating
  $\varphi_a$).
  Similarly, if the schedule begins in mode $\Mm^1_i$ or 
  $\Mm^2_i$, then again we will obtain $z_{i\#}<0$ hitting the obstacle $\neg
  Init$.
  Also, starting in mode $\Mm_{i~halt}$ or $\Mm^c_{halt}$ 
  or $\Mm_{halt}$ will give $x_{i~halt}<0$ or $c<0$ or $w_{halt}<0$,
  respectively, all of which will hit the obstacle $\neg Init$.
  Thus, any safe schedule must begin in mode $\mathcal{I}$. 
  Now we show that each such schedule must spend exactly one unit of time in
  $\mathcal{I}$.   
\begin{enumerate}
\item 
If more than one unit of time is spent at $\mathcal{I}$, then 
$s_0<0$ will violate $\varphi_a$. 
\item Assume that a time $t<1$ is spent at $\mathcal{I}$, and the 
control shifts to any other mode. Then we have $s_0=1-t$ and $w_{0j}=t$ 
on entering that mode. Note that $w_{0j}+s_0=1$. 
If any time is spent at that mode, 
then we will either obtain $s_0>0$ and 
some $x_{ij}\neq 0$ (hits $\neg Mutex(S,X)$ and violates $\varphi_e$) or 
some $w_{jk}, z_{k\#}<0$ (hits $\neg Init$ and violates $\varphi_a$). 
\end{enumerate}
The proof is now complete. \qed
\end{proof}

\begin{lemma}
\label{two}
After spending $1$ time unit in mode $\mathcal{I}$, any safe schedule must
choose mode  $\Mm_0$ followed by mode $\Mm_{0j}$ (where $\ell_j$ is the unique instruction that follows
$\ell_0$ in the two counter machine) both for exactly $1$ time unit. 
\end{lemma}
\begin{proof}
After spending one unit of time in $\mathcal{I}$, we have $s_0=0$ and 
$w_{0j}=1$. We claim that the control will switch to mode $\Mm_0$ from
$\mathcal{I}$.
\begin{enumerate}
\item
  Recall that $R(\Mm_0)(w_{0j})=-1$, and $R(\Mm_0)(x_{0j})=1$, where $j$ is
  the unique index of the instruction $\ell_j$ to which control shifts in the
  two counter machine from $\ell_0$. 
  If control switches to $\Mm_0$, and $0 \leq t \leq 1$ time is spent, then we
  have 
  $w_{0j}=1-t$, $x_{0j}=t$. The resultant points  
  are all safe; note that $w_{0j}+x_{0j}=1$.
\item
  If control switches from $\mathcal{I}$ to some $\Mm_k$, $k \neq 0$, or
  some $\Mm^1_k$ (or $\Mm^2_k$)
  or some $\Mm_{gf}$, 
  then we have
  \begin{itemize}
  \item $R(\Mm_k)(w_{kg})=-1$, and $R(\Mm_k)(x_{kg})=1$ for some $g$, or
  \item $R(\Mm^1_k)(z_{k\#})=-1$, and $R(\Mm_k)(x_{kg})=1$ for some $g$. 
  \item $R(\Mm_{gf})(x_{gf})=-1$.
  \end{itemize}
  If $t>0$ time is spent at $\Mm_k$, or $\Mm^1_k$ ($\Mm^2_k$)
  or $\Mm_{gf}$,  then we obtain 
  $w_{kg}<0$ or  $z_{k\#}<0$ or $x_{gf}=-t<0$, violating the safety
  requirement. 
\end{enumerate}
Next we claim that exactly one unit of time is spent at $\Mm_0$ before control
switches to any other mode. By Lemma \ref{one}, we have $s_0=0, w_{0j}=1$ 
on entering $\Mm_0$.
\begin{enumerate}
\item
  It is easy to see that if time $t>1$ is spent at $\Mm_0$, the obstacle $\neg Init$ is
  hit, since $x_{0j}>1$.
\item
  Assume now that  $t$ time is spent at $\Mm_0$, and control switches to
  some mode. Then we have $w_{0j}=1-t, x_{0j}=t$ on exiting $\Mm_0$. 
  \begin{itemize}
  \item
    If mode $\Mm_{0j}$ is chosen after $\Mm_0$, and a time $t'>0$ is spent at
    $\Mm_{0j}$, then 
    we obtain $x_{0j}=t-t'$, $w_{0j}=1-t$. Also,  we have $w_{jk}=t'>0$ or 
    $z_{j\#}=t'>0$, depending on whether $\ell_j$ is not a zero check 
    instruction having $\ell_k$ as its successor, or $\ell_j$ is a zero check
    instruction. 
    In the former case, we obtain $w_{0j}>0$ and $w_{jk}>0$, violating $\varphi_c$, while
    in the latter case, we obtain $w_{0j}>0$ and $z_{j\#}>0$, again violating
    $\varphi_c$. However, if $t=1$ time is spent at $\Mm_0$, then there is no violation
    since we have $w_{0j}=0, x_{0j}=1-t'$ and exactly one of $w_{jk}=t'$ or
    $z_{j\#}=t'$.  
  \item
    Assume that a time $t=1$ is spent at $\Mm_0$, but a mode other than
    $\Mm_{0j}$ is chosen from $\Mm_0$ and a time $t'>0$ is spent there. 
    \begin{itemize}
    \item
      If $\Mm_k$ is chosen for some $k$, 
      then we obtain  $x_{kg}=t'$. This violates $\varphi_b$ since $x_{kg}=t'$ and
      $x_{0j}=1$.
    \item
      If $\Mm_{kg}$ is chosen for some $k,g \neq 0,j$, then we obtain 
      $w_{gh}=t'$ or $z_{g\#}=t'$ along with $x_{kg}=-t'$. 
      This violates $\varphi_a$ since $x_{kg}=-t' <0$.
    \end{itemize}
  \end{itemize}
\end{enumerate}
Thus, we have seen that starting from $\mathcal{I}$, the control shifts to
$\Mm_0$
and then to $\Mm_{0j}$ in order,  spending exactly one unit of time at
 $\Mm_0$. 
We now argue that the time $t'$ spent at $\Mm_{0j}$ has to be 1.
Assume $t' \neq 1$. Spending $t'$ units of time at $\Mm_{0j}$ results 
in $x_{0j}=1-t'$ and one of $w_{jk}=t'$ or $z_{j\#}=t'$. 
\begin{enumerate}
\item
  If $t' >1$, then we obtain $x_{0j}=1-t' <0$, violating $\varphi_a$.
\item
  Assume $t'<1$, and control switches from $\Mm_{0j}$ to some $\Mm_f$.
  Let $t''$ time be spent at $\Mm_f$. If $f \neq 0$, then we obtain 
  $x_{fg}=t''>0$ for some $g$, and  $x_{0j}=1-t'>0$
  violating $\varphi_b$. If $f=0$, then we obtain $x_{0j}=1-t'+t''$, but
  $w_{0j}=-t''<0$, violating $\varphi_a$.
\item
  Assume $t'<1$, and control switches from $\Mm_{0j}$ to some $\Mm^1_f$ or 
  $\Mm^2_f$.
  Let $t''$ time be spent at $\Mm^1_f$ ($\Mm^2_f$).
  Then we obtain $z_{f\#}=-t''<0$ violating $\varphi_a$.
\item
  Assume $t'<1$, and control switches from $\Mm_{0j}$ to some
  $\Mm_{cd}$. Then we have  one of $z_{d\#}=t''>0$
  or $w_{dh}=t''>0$ along with one of 
  $w_{jk}=t'>0$ or $z_{j\#}=t'>0$, both which violate one of $\varphi_c, \varphi_d$. 
\end{enumerate}
If $t'=1$, then we have $x_{0j}=0$ and 
one of $w_{jk}=1$ or $z_{j\#}=1$.
The proof is now complete. \qed
\end{proof}

Now if $j$ was not a zero check instruction, and has $\ell_k$ 
as the successor of $\ell_j$, then as seen above in Lemma \ref{two} in the case
of $\Mm_0$ and $\Mm_{0j}$, we can show that
the control has to shift to $\Mm_j$ from $\Mm_{0j}$. If $j$ is a zero check
instruction,
then we claim that the control has to switch from $\Mm_{0j}$ to one of 
$\Mm^1_j$ or $\Mm^2_j$. Lemma \ref{three} generalises  
this claim.

\begin{lemma}
\label{three}
If the system is in mode $\Mm_{gf}$ then any safe schedule must pick mode
$\Mm_f$ if $f$ is not a zero check instruction.  
However, if $f$ is a zero check instruction, then the next mode must be either
of mode $\Mm^1_f$ or $\Mm^2_f$.
Any safe schedule must also spent $1$ time unit at modes $\Mm_{gf}$ and at $\Mm_f$,
$\Mm^1_f$, or $\Mm^2_f$ (as is the case).
\end{lemma}
\begin{proof}
  Assume that the system is in mode $\Mm_{gf}$, and assume that it followed a safe
  execution from the starting state.
  We know that $R(\Mm_{gf})(x_{gf})=-1$.
  In this case, upon entering mode $\Mm_{gf}$, we must have $x_{gf}=1$, while all
  variables except $c_1,c_2$ must be $0$.
  There are two cases to consider. 
  \begin{enumerate}
  \item
    $\ell_f \not \in O$. Then $R(\Mm_{gf})(w_{fq})=1$
    for some unique index $q$ corresponding to the successor $\ell_q$ of $\ell_f$.
    Spending $t=1$ here results in $x_{gf}=0$ and $w_{fq}=1$, with no violation to the non-hitting zone.
    As seen in Lemma \ref{two} for the case of $\Mm_{0j}$, a time $t'$ spent at
    $\Mm_{gf}$ is safe 
    iff $t'=1$ ($t'>1$ violates safety immediately, while a switch in control 
    from $\Mm_{gf}$ with $t'<1$ to any mode
    disallows spending time at the new mode).
  \item
    $\ell_f \in O$.
    Then $R(\Mm_{gf})(z_{f\#})=1$.
  Spending $t$ unit of time at 
  $\Mm_{gf}$ results in $x_{gf}=1-t$ and $z_{f\#}=t$.
  \begin{enumerate}
\item If $t>1$, then we obtain $x_{gf}<0$ violating $\varphi_a$.   
 \item Assume $t<1$, and control switches 
 out of $\Mm_{gf}$. 
 \begin{itemize}
 \item If the next mode chosen is some $\Mm_k$, and $t'>0$ units of time spent,
   then we obtain 
 $w_{k*}=-t'<0$ violating $\varphi_a$.
\item If the next mode chosen is some $\Mm^1_k$ or
$\Mm^2_k$ with $k \neq f$, and $t'$ units of time spent there, then we obtain
$z_{k\#}=-t'<0$ violating $\varphi_a$.
\item If the next mode chosen 
is some $\Mm_{cd}$, then we obtain $x_{cd}=-t'<0$
 violating $\varphi_a$.
 \item If the next mode chosen is 
 $\Mm_f^1$ or $\Mm_f^2$, and $t'>0$ units of time spent, then we obtain 
 $z_{f\#}=t-t'$.  However, we also obtain  $x_{f*}=t'>0$ and 
 $x_{gf}=1-t >0$ violating $\varphi_b$. 
 \end{itemize}
\item Assume $t=1$ unit of time is spent at $\Mm_{gf}$ and control switches out
  of $\Mm_{gf}$. We then have 
$x_{gf}=0$ and $z_{f\#}=1$.
\begin{itemize}
 \item If the next mode chosen is some $\Mm_k$, and $t'>0$ units of time spent,
   then we obtain 
 $w_{k*}=-t'<0$ violating $\varphi_a$.
 \item If the next mode chosen is $\Mm_j^1$ or $\Mm_j^2$, $j \neq f$, and $t'>0$
   units of time 
 spent, then we obtain $z_{j\#}=-t'<0$ violating $\varphi_a$.
\item If the next mode chosen is some $\Mm_{cd}$ and 
$t'>0$ units of time spent, then we obtain 
 $x_{cd}=-t'<0$, violating $\varphi_a$.
\item  If the next mode chosen is 
 $\Mm_f^1$ or $\Mm_f^2$, and $t'>0$ units of time spent, 
 we obtain $z_{f\#}=1-t'$, along with $x_{f*}=t'>0$.
 If $\ell_f$ is the instruction ``if $c_1>0$, goto $f_1$, else goto $f_2$'',
 then 
 * is either $f_1$ or $f_2$. There is no violation to safety as long as  $c_1>0$
  and $\Mm^1_f$ is chosen, or $c_1=0$ and $\Mm_f^2$ is chosen when $t'\leq 1$.
  If $t'=1$, then we obtain $z_{f\#}=0$ and $x_{f*}=1$.
    After spending $t'=1$ unit of time at $\Mm_f^1$ or $\Mm_f^2$, 
   assume the control switches to 
  $\Mm_{f*}$, and a time $t''$ is spent there.  As seen in Lemma \ref{two}, 
  it can be shown that there is no violation to safety as long as $t''\leq 1$. 
In particular, it can be shown that if $t''<1$, and control switches 
out of $\Mm_{f*}$, safety is violated. 
\end{itemize}
\end{enumerate}
\end{enumerate}
%The proof is now complete.\qed
\end{proof}

Thus, it can be seen that some obstacle is hit if a time other than one  is
spent  at any mode, or if a mode violating the order of instructions in the two
 counter machine is chosen.  
Lemmas \ref{one}, \ref{two} and \ref{three} prove this.  
Assume now that the mode switching happens respecting  the 
instruction flow in the two counter machine, and one unit of time is spent at
each mode. It can be seen that 
two counter machine halts iff some mode $\Mm_{i~halt}$ is reached. 
After spending one unit of time at $\Mm_{i~halt}$, we obtain $w_{halt}=1$, and
all 
$x$ variables 0. Note that by condition $\varphi_f$,  
no $x$ variables other than $x_{i~halt}$ can be non-zero when 
$w_{halt}>0$. 

\begin{lemma}
\label{last}
Any safe schedule, upon entering mode $\Mm_{i~halt}$, must spend exactly one unit of time.
\end{lemma}
\begin{proof}
On entering $\Mm_{i~halt}$, we have $x_{i~halt}=1$ and $w_{halt}=0$.
Assume $t<1$ time is spent at $\Mm_{i~halt}$ and a mode change happens. Then we
have 
$x_{i~halt}=1-t$ and $w_{halt}=t$.
\begin{itemize}
\item If we move to any $\Mm_k, \Mm^1_k, \Mm^2_k$ or $\Mm_{cd}$, and elapse a
  time $t'>0$, we 
will have a safety violation due to some variable becoming negative ($w_{k*}$ in
the case of $\Mm_k$, 
$x_{cd}$ in the case of $\Mm_{cd}$ and $z_{k\#}$ in the case of 
$\Mm^1_k, \Mm^2_k$).
\item Assume that we move to $\Mm_{halt}$ and spend $t'>0$ time there. 
Then we obtain $w_{halt}=t-t'$. This violates $\varphi_g$  since we have $x_{i~halt}+w_{halt} \neq 1$.
Moving to $\Mm^{c_1}_{halt}, \Mm^{c_2}_{halt}$ also violates safety for the same
reason.
\end{itemize}
However, if $t=1$, then we have $w_{halt}=1$
and all other $x,n$ variables are 0. Moving to any mode other 
$\Mm_{halt}$ or 
$\Mm^{c_1}_{halt}, \Mm^{c_2}_{halt}$ will violate $\varphi_f$. 
\end{proof}

\begin{lemma}[Correctness]
\label{really-last}
The target point $(0,0,0,\dots,1)$ is safely reachable from the starting point $(1,0,0,
\dots,0)$ in multi-mode system $\Hh_\Aa$  iff the two counter machine $\Aa$ halts. 
\end{lemma}
\begin{proof}
Assume that the two counter machine halts. Then starting from the initial
instruction, 
we reach the instruction $\ell_{halt}$. From the above lemmas, and the
construction 
of the MMS, we know that starting from the initial mode $\mathcal{I}$, we will 
reach a unique mode $\Mm_{i~halt}$, spending one unit of time at all 
the intermediate modes. From Lemma \ref{last}, we also know that a time of one
unit is spent at 
$\Mm_{i~halt}$, and that the only safe modes to goto from here are    
$\Mm_{halt}$ or $\Mm^{c_1}_{halt}, \Mm^{c_2}_{halt}$. 
The use of modes 
$\Mm^{c_1}_{halt}, \Mm^{c_2}_{halt}$
is just to obtain $c_1=c_2=0$. Notice that $w_{halt}$ stays non-zero and grows
in these modes. 
Once we achieve $c_1=c_2=0$, then we can visit $\Mm_{halt}$ and 
obtain $w_{halt}=1$. Notice that when this happens, we will have 
all variables other than $w_{halt}$ as 0. By our safety conditions, 
no $x, w$ or $z$ variable can stay non-zero when $w_{halt}>0$.
Even if we obtain $w_{halt}=0$ by staying at $\Mm_{halt}$, we cannot visit any
other mode, since 
atleast one variable will become negative and violate safety. Our starting
configuration 
with $s_0=1$ ensured that we could start from $\mathcal{I}$ and continue in a
safe manner.

The converse,  that is, any safe schedule starting from $(1,0,0, \dots, 0)$
and reaching $(0,0,0, \dots, 1)$ is possible only when the two counter machine
halts by the construction of the MMS.
\qed
\end{proof}

\subsection{Undecidability: Example of the reduction} 
\label{app:undec-eg}

Consider an example of a two counter machine with counters $c_1,c_2$ and
  the following instruction set.
  \begin{itemize}
    \item $\ell_0 : c_1 := c_1+1$;  goto  $\ell_1$
    \item $\ell_1 : c_1 := c_1-1$;  goto  $\ell_2$
    \item $\ell_2$ : if $(c_2{>}0)$ then goto $\ell_3$;
      else goto $\ell_0$
    \item $\ell_3:$ HALT.
  \end{itemize}
  Note that this machine does not halt.
  We now describe a multi-mode system that simulates this two counter machine. 
  The modes and variables are as follows.
 \begin{enumerate}
 \item Variables $\{c_1, c_2\}$ correspond to the two counters, $\{w_{01}, w_{12},z_{2\#}\}$  correspond to instructions 
 $\ell_0,\ell_1,\ell_2, \ell_3$ and the switches between instructions, and variable $w_3$ corresponds to the halt  
 instruction. We also have variables $s_0$ and $\{x_{ij} \mid 0 \leq i,j \leq 3\}$.   
 \item Mode $\mathcal{I}$: $R(\mathcal{I})(s_0)=-1$ and 
 $R(\mathcal{I})(w_{01})=1$ and other variables have rate 0.
 \item Modes $\Mm_0, \Mm_1, \Mm_2^1, \Mm_2^2, \Mm_{01}, \Mm_{12}, \Mm_{20}$ with rates
\begin{itemize}
\item  $R(\Mm_0)(w_{01})=-1, R(\Mm_0)(x_{01})=R(\Mm_0)(c_1)=1$, and $R(\Mm_0)(v)=0$ for all other variables $v$
\item  $R(\Mm_1)(w_{12})=-1, R(\Mm_1)(x_{12})=1, R(\Mm_1)(c_1)=-1$, and $R(\Mm_1)(v)=0$ for all other variables $v$
\item  $R(\Mm^1_2)(z_{2\#})=-1=R(\Mm^2_2)(z_{2\#}),
 R(\Mm_2^1)(x_{23})=1=R(\Mm_2^2)(x_{20})$. All other variables have rate 0 in modes 
 $\Mm_2^1,\Mm_2^2$.
 \item $R(\Mm_{12})(z_{2\#})=1$, and $R(\Mm_{12})(v)=0$ for all other variables $v$
 \item $R(\Mm_{01})(x_{01})=-1$, $R(\Mm_{01})(w_{12})=1$,  and $R(\Mm_{01})(v)=0$ for other $v$ 
 \item $R(\Mm_{20})(x_{20})=-1$, $R(\Mm_{20})(w_{01})=1$, and $R(\Mm_{20})(v)=0$ for  other $v$ 
 \end{itemize}
\item Mode $\Mm_{23}$ with  $R(\Mm_{23})(x_{23})=-1,R(\Mm_{23})(w_3)=1$,  $R(\Mm_{23})(v)=0$ for other $v$.  Modes $\Mm_3, \Mm_3^{c_1}, \Mm_3^{c_2}$ with 
$R(\Mm_3)(w_3)=-1$ and $R(\Mm_3)(v)=0$ for all other variables $v$; 
$R(\Mm^{c_i}_3)(c_i)=-1, R(\Mm^{c_i}_3)(w_3)=1$. 
 \end{enumerate}

The safety set is given by the conjunction of conditions (1)-(7).\\
(1) $0 \leq w_{01}, w_{12},x_{ij},z_{2\#},s_0 \leq 1$, $0 \leq w_3, c_1, c_2$, 
(2) At any point, if some $x_{ij}$ is non-negative, then all the other $x_{kl}$ variables are 0, $i \neq k, j \neq l$.
(3) At any point, if some $w_{ij}$ is non-negative, then all other $w_{kl}$ are zero, $i \neq k, j \neq l$ and 
$z_{2\#}=0$.
(4)At any point, if $z_{2\#}>0$, then all the $w_{ij}$ are 0.
(5) At any point, if $s_0>0$, then all $x_{ij}=0$.
(6) At any point, if $w_3 >0$, then variables $x_{i3}$ are 0, and $c_1,c_2=0$.
(7)At any point, if $x_{i3}>0$, then $x_{i3}+w_3=1$.
The obstacles are hence, the complement of this conjunction. 

Starting from $s_0=1$ and $v=0$ for $v \neq s_0$, a safe computation 
must start from $\mathcal{I}$, then visit in order modes $\Mm_0, \Mm_{01}, \Mm_1, \Mm_{12}, \Mm_2^1, \Mm_{20}$, and repeat 
this sequence spending one unit in each mode.  This will not reach $\Mm_3$. 

We start from the point $s_0=1$ and $v=0$ for all variables $v \neq s_0$. The safe set of values for $s_0$ is $[0,1]$. As seen in Lemma \ref{one}, computation starts in mode $\mathcal{I}$, and one unit of time is spent there.
This results in $s_0=0, w_{01}=1$, and $v=0$ for all other variables. The control then shifts to modes 
$\Mm_0, \Mm_{01}$ in order. One unit of time is spent in $\Mm_0$, and we obtain 
$w_{01}=0, x_{01}=c_1=1$, and $v=0$ for all other variables $v$. 
This is followed by spending one unit of time in $\Mm_{01}$, obtaining 
$x_{01}=0$, $w_{12}=1$, $c_1=1$ and $v=0$. The control shifts from 
  $\Mm_{01}$ to $\Mm_1$, where one unit of time is spent. This 
  results in $w_{12}=0, x_{12}=1$ and $c_1=0$, and $v=0$ 
  for all other variables $v$. Control then shifts to $\Mm_{12}$, where one unit of time is spent.
  This results in obtaining $z_{2\#}=1$ and $c_1=1$, and $v=0$ for all other variables. Control then shifts 
    to $\Mm_2^1$, and one unit of time is spent, resulting in $z_{2\#}=0, x_{20}=1, c_1=0$. This is continued by visiting      mode $\Mm_{20}$, and we obtain $x_{20}=0, w_{01}=1$, after spending one unit of time in $\Mm_{20}$.  
     Any other sequence of visiting modes, or spending times other than 1 in the visited modes 
   will result in hitting an obstacle.

   The computation now proceeds to modes $\Mm_0, \Mm_{01}, \Mm_1, \Mm_{12}, \Mm_2^1, \Mm_{20}$ 
 in a loop.  Note that the halt mode $\Mm_3$ is never reached.  

\end{document}